\documentclass[twoside]{IEEEtran}

\IEEEoverridecommandlockouts
\usepackage{amsmath}
\usepackage{amssymb}
\usepackage{amsthm}
\usepackage{enumerate}

\usepackage{epsfig}
\usepackage{subfigure}
\usepackage{psfrag}
\usepackage{cite}
\usepackage{stfloats}
\usepackage[mathscr]{euscript}
\usepackage{hyperref}
\hypersetup{backref, colorlinks=true, linkcolor=blue}
\usepackage{microtype} 

\graphicspath{{figures/}}

\newcommand{\matn}{\ensuremath{\mathcal{N}}}

\ifx\eqref\undefined
	\newcommand{\eqref}[1]{~(\ref{#1})}
\fi
\ifx\mod\undefined
	\def\mod{\mathop{\rm mod}}
\fi

\newcommand{\bigo}[1]{O\left(#1\right)}
\newcommand{\smallo}[1]{o\left(#1\right)}
\newcommand{\Qinv}[1]{Q^{-1}\left(#1\right)}
\newcommand{\Var}[1]{\mathrm{Var}\left[#1\right]}

\newcommand{\E}[1]{\mathbb E\left[#1\right]}
\newcommand{\Prob}[1]{\mathbb P\left[#1\right]}

\newcommand{\1}[1]{1\left\{#1\right\}}
\newcommand{\beq}{\begin{equation}}
\newcommand{\eeq}{\end{equation}}

\def\exp{\mathop{\rm exp}}

\def\EE{\mathbb{E}\,}

\def\PP{\mathbb{P}}

\def\eqdef{\stackrel{\triangle}{=}}

\def\unifto{\mathop{{\mskip 3mu plus 2mu minus 1mu%
	\setbox0=\hbox{$\mathchar"3221$}%
	\raise.6ex\copy0\kern-\wd0%
	\lower0.5ex\hbox{$\mathchar"3221$}}\mskip 3mu plus 2mu minus 1mu}}

\ifx\lesssim\undefined
\def\simleq{{{\mskip 3mu plus 2mu minus 1mu%
	\setbox0=\hbox{$\mathchar"013C$}%
	\raise.2ex\copy0\kern-\wd0%
	\lower0.9ex\hbox{$\mathchar"0218$}}\mskip 3mu plus 2mu minus 1mu}}
\else
\def\simleq{\lesssim}
\fi

\ifx\gtrsim\undefined
\def\simgeq{{{\mskip 3mu plus 2mu minus 1mu%
	\setbox0=\hbox{$\mathchar"013E$}%
	\raise.2ex\copy0\kern-\wd0%
	\lower0.9ex\hbox{$\mathchar"0218$}}\mskip 3mu plus 2mu minus 1mu}}
\else
\def\simgeq{\gtrsim}
\fi

\newif\ifmapx
{\catcode`/=0 \catcode`\\=12/gdef/mkillslash\#1{#1}}
\edef\jobnametmp{\expandafter\string\csname varrate_apx\endcsname}
\edef\jobnameapx{\expandafter\mkillslash\jobnametmp}
\edef\jobnameexpand{\jobname}
\ifx\jobnameexpand\jobnameapx
\mapxtrue
\else
\mapxfalse
\fi

\long\def\apxonly#1{\ifmapx{\color{blue}#1}\fi}

\newtheorem{thm}{Theorem}

\newtheorem{lemma}{Lemma}

\newtheorem{defn}{Definition}

\theoremstyle{remark}
\newtheorem{remark}{Remark}
\newtheorem*{example}{Example}

\psfrag{k}{\footnotesize{$k$}}
\psfrag{R}{\footnotesize{Rate, bits per bit}}
\psfrag{Exact}{\footnotesize{$\frac{1}{k} L_{S^k}^\star ( \epsilon)$} }
\psfrag{Approx}{\footnotesize{Approximation \eqref{Reps2order}}}
\psfrag{LowerHsv}{\footnotesize{Lower bound \eqref{palmdrive}}}
\psfrag{LowerH}{\footnotesize{Lower bound \eqref{eq:itt0} }}
\psfrag{UpperH}{\footnotesize{Upper bound \eqref{eq:itt1}}}
\psfrag{UpperI}{\footnotesize{Upper bound \eqref{eq:lmb4}}}

\psfrag{I}{\footnotesize{bits per bit }}
\psfrag{Iexact}{\footnotesize{$\frac 1 k \mathbb H(S^k, \epsilon)$ \eqref{jamon}}}
\psfrag{Ilower}{\footnotesize{Lower bound \eqref{eq:lmb5}}}
\psfrag{Iupper}{\footnotesize{Upper bound \eqref{eq:lmb6}}}

\title{\LARGE{Variable-length compression allowing errors}}

\author{
Victoria Kostina, ~\IEEEmembership{Member,~IEEE,}
~Yury Polyanskiy,  ~\IEEEmembership{Senior Member,~IEEE,}
~Sergio Verd\'u,  ~\IEEEmembership{Fellow,~IEEE,}

\thanks{
This work was supported in part by the Center for Science of Information
(CSoI), an NSF Science and Technology Center, under Grant CCF-0939370. This paper was presented in part at ISIT 2014 \cite{kostina2014varrateISIT}. 

V. Kostina is with California Institute of Technology, Y. Polyanskiy is with MIT, S. Verd\'u  is with Princeton University. (e-mail: \href{mailto:vkostina@caltech.edu}{vkostina@caltech.edu}; 
\href{mailto:yp@mit.edu}{yp@mit.edu};
\href{mailto:verdu@princeton.edu}{verdu@princeton.edu}.)
}
}

\begin{document}
\maketitle
\begin{abstract} 
This paper studies the fundamental limits of the minimum average length of lossless and lossy variable-length compression, allowing a nonzero error probability $\epsilon$, for lossless compression. 
We give non-asymptotic bounds on the minimum average length in terms of Erokhin's rate-distortion function
and we use those bounds to obtain a Gaussian approximation on the speed of approach to the limit which is quite accurate for all but small blocklengths:
$$(1 - \epsilon) k H(\mathsf S) -   \sqrt{\frac{k V(\mathsf S)}{2 \pi} } e^{- \frac {(\Qinv{\epsilon})^2} 2 }$$ where $\Qinv{\cdot}$ is the functional inverse of the standard Gaussian complementary cdf, and  $V(\mathsf S)$ is the source dispersion. A nonzero error probability thus not only reduces the asymptotically achievable rate by a factor of $1 - \epsilon$, but this asymptotic limit is approached from {\it below}, i.e. larger source dispersions and shorter blocklengths are beneficial. Variable-length lossy compression under an excess distortion constraint is shown to exhibit similar properties. 
\end{abstract}


\section{Introduction and summary of results}

Let $S$ be a discrete random variable  to be compressed into a variable-length binary
string. We denote the set of all binary strings (including the empty string) by $ \left\{0, 1\right\}^\star$
and the length of a string $a \in  \left\{0, 1\right\}^\star$ by $\ell(a)$. 
 The codes considered in this paper fall under the following paradigm. 
\begin{defn}[$(L, \epsilon)$ code] \label{defn:Leps}
 A variable length $(L, \epsilon)$ 
 code for source $S$ defined on a finite or countably infinite
 alphabet $\mathcal M$ is a pair of possibly random transformations $P_{W|S} \colon \mathcal M \mapsto \left\{0, 1\right\}^\star$
 and $P_{\hat S|W} \colon \left\{0, 1\right\}^\star \mapsto \mathcal M$ such that\footnote{Note that $L$ need not be an integer.}
 \begin{align}
\Prob{S \neq \hat S } &\leq \epsilon \label{eq:fidelity}\\
\E{\ell( W ) } &\leq L
\end{align}
The corresponding fundamental limit is
\begin{equation}
L_S^\star(\epsilon) \triangleq \inf  \left\{L \colon \exists \text{ an $(L, \epsilon)$ code} \right\} \label{leps(S)}
\end{equation}
\end{defn}
Lifting the prefix condition in variable-length coding is discussed in \cite{szpankowski2011minimum,verdu2014lossless}.
In particular, in the zero-error case we have \cite{alon1994lower, AW72}
\begin{align}\label{eq:alon2}
		H(S) - \log_2 (H(S)+1) - \log_2 e &\le L_S^\star(0) \\
		&\le H(S)\,, \label{eq:wyner}
\end{align}	
while
\cite{szpankowski2011minimum} shows that in the i.i.d. case (with a non-lattice distribution $P_{\mathsf S}$, otherwise $o(1)$ becomes $O(1)$)
\begin{equation}\label{SVformula}
 L_{ S^k}^\star (0) =  k\, H(\mathsf S) - \frac12  \log_2 \left( 8 \pi e V( \mathsf S) k \right) + o(1)
\end{equation}
where $V( \mathsf S)$ is the \textit{varentropy} of $P_\mathsf S$, namely the variance of the information 
\begin{align}
\label{eq:info}
\imath_{\mathsf S} (\mathsf S) = \log_2 \frac{1}{P_{\mathsf S} (\mathsf S)}.
\end{align}

Under the rubric of  ``weak variable-length source coding," T. S. Han \cite{han2000weakvariable}, \cite[Section 1.8]{han2003information} 
considers the asymptotic fixed-to-variable   ($\mathcal{M} = \mathcal{S}^k$) almost-lossless version
of the foregoing setup with vanishing error probability and prefix encoders.
Among other results, Han showed that the minimum average length $L_{S^k}(\epsilon)$ of prefix-free encoding of a stationary ergodic source with entropy rate $H$ behaves as 
\begin{equation}
\lim_{\epsilon \to 0} \lim_{k \to \infty} \frac 1 k L_{ S^k} (\epsilon) = H .
\end{equation}
Koga and Yamamoto \cite{koga2005weakfv} characterized asymptotically achievable rates of variable-length prefix codes with non-vanishing error probability and, in particular, showed that for finite alphabet i.i.d. sources
with distribution $P_{\mathsf S}$,
\begin{equation}\label{kogaformula}
\lim_{k \to \infty} \frac 1 k L_{ S^k} (\epsilon) = (1 - \epsilon) H(\mathsf S). 
\end{equation}
The benefit of variable length vs. fixed length in the case of given $\epsilon$ is clear from \eqref{kogaformula}: indeed, the latter satisfies a strong converse and therefore any rate below the entropy is fatal. 
Allowing both nonzero error and variable-length  coding
is interesting not only conceptually but on account on several important generalizations.
For example, the variable-length counterpart of Slepian-Wolf coding 
considered e.g. in \cite{kimura2004weak} is particularly relevant 
in universal settings, and has a radically different (and practically uninteresting) zero-error version.
Another substantive important generalization where nonzero error is inevitable is variable-length joint source-channel coding without or with feedback.
For the latter, Polyanskiy et al. \cite{polyanskiy2011feedback} showed that allowing a nonzero error probability boosts the $\epsilon$-capacity of the channel, while matching the transmission length to channel conditions accelerates the rate of approach to that asymptotic limit. The use of nonzero error compressors is 
also of interest in hashing \cite{kogaHash}. 


The purpose of Section \ref{sec:lossless} is to give non-asymptotic bounds on the fundamental limit \eqref{leps(S)}, and to apply those bounds to
analyze the speed of approach to the limit in \eqref{kogaformula}, which also holds without the prefix condition. Specifically, we show that (cf. \eqref{eq:alon2}--\eqref{eq:wyner}) 
\begin{align}\label{eq:approx}
		L_S^\star(\epsilon) &=  \mathbb{H} (S,\epsilon)  + \bigo{\log_2 H(S)}\\
		&= \EE[ \left \langle \imath_S(S) \right \rangle_\epsilon] + \bigo{\log_2 H(S)}
\end{align}
where
\begin{equation}\label{eq:heps_def2}
 \mathbb{H} (S,\epsilon) \eqdef \min_{ \substack{ 
			P_{Z | S }
			 \colon 
			 \\
			  \Prob{S \neq Z} \leq \epsilon
			 }
			 } I(S; Z)  \\
\end{equation}
is Erokhin's function \cite{erokhin1958epsilon}, and the $\epsilon$-cutoff random transformation acting on a real-valued random variable $X$ is defined as
\begin{equation}
\langle X \rangle_\epsilon \triangleq 
\begin{cases}
X & X < \eta\\
\eta & X = \eta ~(\text{w. p. } 1 - \alpha )\\
0 & X = \eta ~(\text{w. p. } \alpha )\\
0 & \text{otherwise} 
\end{cases} \label{Xeps}
\end{equation}
where $\eta \in \mathbb R$ and $\alpha \in [0, 1)$ are determined from
\begin{equation}
\Prob{ X  > \eta} + \alpha\, \Prob{X = \eta}  = \epsilon. \label{epscutoff}
\end{equation}
While $\eta$ and $\alpha$ satisfying \eqref{epscutoff} are not unique in general, any such pair defines the same $\langle X \rangle_\epsilon$ up to almost-sure equivalence.

The code that achieves \eqref{eq:approx} essentially discards ``rich'' source realizations with $\imath_S(S) > \eta$ and encodes the rest losslessly assigning them in the order of decreasing probabilities to the elements of $\left\{0, 1\right\}^\star$ ordered lexicographically. 

For memoryless sources with $S_i \sim \mathsf S$ we show that the speed of approach to the limit in \eqref{kogaformula} is given by the following result.
\begin{align}
 \left.\begin{aligned}
        & L_{S^k}^\star(\epsilon)\\
        & \mathbb{H} (S^k, \epsilon)\\
        &\E{ \left \langle \imath_{S^k}(S^k) \right \rangle_\epsilon}
       \end{aligned}
 \right\}
 &= (1 - \epsilon) k H(\mathsf S) -   \sqrt{\frac{k V(\mathsf S)}{2 \pi} } e^{- \frac { (\Qinv{\epsilon})^2} 2 } 
 \notag\\
 & + \bigo{\log k}  \label{Reps2orderIntro} 
\end{align}

To gain some insight into the form of \eqref{Reps2orderIntro}, note that if the source is memoryless, the information in $S^k$ is a sum of i.i.d. random variables, and by the central limit theorem
\begin{align}
\imath_{S^k} (S^k) &=  \sum_{i = 1}^k \imath_{\mathsf S}(S_i) \label{eq:infosum}\\
&\stackrel{d}{\approx} \mathcal N \left( k H(\mathsf S), k V(\mathsf S) \right)
\end{align}
while for Gaussian $X$
\begin{equation}
 \E{ \left \langle X \right \rangle_\epsilon } = (1 - \epsilon) \E{ X} -   \sqrt{\frac{\Var{ X}}{2 \pi} } e^{- \frac { (\Qinv{\epsilon})^2} 2 }
\end{equation}

Our result in \eqref{Reps2orderIntro} underlines that not only does $\epsilon > 0$ allow for a $(1-\epsilon)$ reduction in asymptotic rate (as found in \cite{koga2005weakfv}), but, in contrast to \cite{strassen1962asymptotische,polyanskiy2010channel,kostina2011fixed,kosut2014asymptotics}, 
larger source dispersion is beneficial. This curious property is further discussed in Section~\ref{sec:discussion}.

\par
In Section~\ref{sec:lossy}, we generalize the setting to allow a general distortion measure in lieu of the Hamming distortion in \eqref{eq:fidelity}.
More precisely, we replace \eqref{eq:fidelity} by  the excess probability constraint $\Prob{\mathsf d \left( S, Z \right) > d} \leq \epsilon$.
In this setting, refined asymptotics of minimum achievable lengths of variable-length lossy prefix codes  almost surely operating at distortion $d$ was studied in \cite{kontoyiannis2000pointwise} (pointwise convergence) and in \cite{zhang1997redundancy,yang1999redundancy} (convergence in mean). 
 Our main result in the lossy case is that \eqref{Reps2orderIntro} generalizes simply by replacing $H(\mathsf S)$ and $V(\mathsf S)$ by the corresponding rate-distortion and rate-dispersion functions, replacing Erokhin's function by
\begin{equation}
 {\mathbb R}_{S}(d, \epsilon) \triangleq  \min_{ \substack{ P_{Z | S } \colon \\ \Prob{ \mathsf d( S, Z) > d} \leq \epsilon}} I(S; Z) \label{RR(eps)},
\end{equation} 
and replacing the $\epsilon$-cutoff of information by that of $\mathsf d$-tilted information \cite{kostina2011fixed}, $\left \langle \jmath_S(S, d) \right \rangle_\epsilon$. Moreover, we show that the $(d, \epsilon)$-entropy of $S^k$ \cite{posner1967epsilonentropy} admits the same asymptotic expansion. If only deterministic encoding and decoding operations are allowed, the basic bounds \eqref{eq:alon2}, \eqref{eq:wyner} generalize simply by replacing the entropy by the $\left(d, \epsilon\right)$-entropy of $S$.   In both the almost-lossless and the lossy case we show that the optimal code is ``almost deterministic'' in the sense that randomization is performed on at most one codeword of the codebook. Enforcing deterministic encoding and decoding operations ensues a penalty of at most $0.531$ bits on average achievable length.

\section{Almost lossless variable length compression}
\label{sec:lossless}

\subsection{Optimal code}\label{sec:optimal}
In the zero-error case the optimum variable-length compressor without prefix constraints $\mathsf f_S^\star$ is known  explicitly (e.g. \cite{leung1978some,alon1994lower})\footnote{The construction in \cite{leung1978some} omits the empty string.}: 
a deterministic mapping that assigns the elements in $\mathcal M$ (labeled without loss of generality as the positive integers) ordered in decreasing probabilities to $\left\{0, 1\right\}^\star$ ordered lexicographically. 
The decoder is just the inverse  of this injective mapping.
This code is optimal in the strong stochastic sense that the cumulative distribution function of the length of any other code cannot lie above that achieved with $\mathsf f_S^\star$. The length function of the optimum code is \cite{alon1994lower}:
\begin{align} 
\ell (\mathsf f_S^\star (m)) = \lfloor \log_2 m \rfloor. 
\end{align}

Note that the ordering $P_{S}(1) \geq P_S(2) \geq \ldots$ implies 
\begin{equation}\label{eq:lmb0}
		\lfloor \log_2 m \rfloor \le \imath_S(m).
\end{equation}

In order to generalize this code to the nonzero-error setting, we take advantage of the fact that in our setting, error detection is not required at the decoder.
This allows us to retain the same decoder as in the zero-error case. As far as the encoder is concerned, to save on length on a given set of  realizations 
which we are willing to fail to recover correctly,  it is optimal to 
assign them all to $\varnothing$. Moreover, since we have the freedom to choose the set that we want
to recover correctly (subject to a constraint 
on its probability $\geq 1 - \epsilon$) it is optimal to include all the most likely realizations (whose encodings according to $\mathsf f_S^\star$ are shortest). 
If we are fortunate enough that $\epsilon$ is such that $\sum_{m =1}^{M} P_S ( m) = 1 - \epsilon$ for some $M$, then the optimal code is
$\mathsf f ( m ) = \mathsf f_S^\star (m )$, if $m = 1, \ldots, M$ and $\mathsf f ( m ) = \varnothing$, if $m > M$.\footnote{Jelinek \cite[Sec 3.4]{jelinek1968probabilistic} provided an asymptotic analysis of a scheme in which a vanishing portion of the least likely source outcomes is mapped to the same codeword, while the rest of the source outcomes are encoded losslessly.}

Formally, for a given encoder $P_{W|S}$, the optimal decoder is always
	deterministic and we denote it by $\mathsf g$. 
	Consider $w_0 \in \{0,1\}^\star \setminus \varnothing$ and source realization $m$
	with $P_{W|S = m}(w_0)>0$. If $\mathsf g(w_0) \neq m$, the average length can be decreased, without affecting
	the probability of error, by setting $P_{W|S = m}(w_0) = 0$ and adjusting $P_{W|S = m}(\varnothing)$ accordingly. This
	argument implies that the optimal encoder has at most one source
	realization $m$ mapping to each $w_0 \neq \varnothing$.   Next, let $m_0 = \mathsf g(\varnothing)$ and  by a similar argument conclude that $P_{W|S = m_0}(\varnothing)=1$. But then,
	interchanging $m_0$ and $1$ leads to the same or better probability of error and
	shorter average length, which implies that the optimal encoder maps $1$ to $\varnothing$. Continuing in the same manner for $m_0 = \mathsf g(0),  \mathsf g(1), \ldots, \mathsf g(\mathsf f^\star_S(M))$, we conclude that 
	the optimal code maps $\mathsf f ( m ) = \mathsf f_S^\star (m )$, $m = 1, \ldots, M$.  Finally, assigning the remaining source outcomes whose total mass is $\epsilon$ to $\varnothing$ shortens the average length without affecting the error probability, so $\mathsf f ( m ) = \varnothing$, $m > M$ is optimal.

We proceed to describe an optimum construction that holds without the foregoing fortuitous choice of $\epsilon$. Let $M$ be the smallest integer such that $\sum_{m =1}^{M} P_S ( m) \geq 1 - \epsilon$, let $\eta = \lfloor \log_2 M \rfloor$, and let $\mathsf f ( m ) = \mathsf f_S^\star (m )$, if $\lfloor \log_2 m \rfloor < \eta$ and $\mathsf f ( m ) = \varnothing$, if $\lfloor \log_2 m \rfloor > \eta$, and assign the outcomes with $\lfloor \log_2 m \rfloor = \eta$
 to $\varnothing$ with probability $\alpha$ and to the lossless encoding $\mathsf f_S^\star (m )$ with probability $1 - \alpha$, which is chosen so that\footnote{It does not matter how the encoder implements randomization on the boundary as long as conditioned on $\lfloor \log_2 S \rfloor = \eta $, the probability that $S$ is mapped to $\varnothing$ is $\alpha$. In the deterministic code with the fortuitous choice of $\epsilon$ described above, $\alpha$  is the ratio of the probabilities
of the sets $\{m \in \mathcal M\colon  m > M ,  \lfloor \log_2 m \rfloor = \eta \}$ to $\{m\in \mathcal M\colon  \lfloor \log_2 m \rfloor = \eta \}$. }
 \begin{align}
\epsilon &= \alpha \sum_{\stackrel{ m \in \mathcal{M}:} {\lfloor \log_2 m \rfloor = \eta}} P_S (m)  +  \sum_{\stackrel{m \in \mathcal{M}:}{ \lfloor \log_2 m \rfloor > \eta}} P_S (m) \\
&= \E{ \varepsilon^\star (S) } \label{ladygaga}
\end{align}
where
\begin{equation}
 \varepsilon^\star(m) = 
\begin{cases}
0 &  \ell (\mathsf f_S^\star (m)) < \eta \\
\alpha &  \ell (\mathsf f_S^\star (m)) = \eta\\
1 &  \ell (\mathsf f_S^\star (m)) > \eta
\end{cases}
\label{epsstar}
\end{equation}

We have shown that the output of the optimal encoder has structure\footnote{If error detection is required and $\epsilon \geq P_{S}(1)$, then $\mathsf f_S^\star(m)$ in the right side of
\eqref{festar} is replaced by $\mathsf f_S^\star(m+1)$. Similarly, if error detection is required and $P_{S}(j) > \epsilon
\geq P_{S}(j+1)$, $\mathsf f_S^\star(m)$ in the right side of \eqref{festar} is replaced by $\mathsf f_S^\star(m+1)$ as long as $m
\geq j$, and $\varnothing$ in the right side of \eqref{festar} is replaced by $\mathsf f_S^\star(j)$. }
 \begin{equation}
W (m ) =   
\begin{cases}
	\mathsf \mathsf f_S^\star (m) & \langle \ell (\mathsf f_S^\star (m))  \rangle_\epsilon > 0 \\
 	\varnothing &\text{otherwise}
\end{cases}
\label{festar}
\end{equation}
and that the minimum average length is given by
 	\begin{align}
			L_S^\star(\epsilon) &= \EE[ \left \langle \ell (\mathsf f_S^\star (S)) \right \rangle_\epsilon ]  \label{eq:lstar}\\
			&= L_S^\star(0) - \max_{\varepsilon(\cdot): \EE[\varepsilon(S)] \le \epsilon} \EE[\varepsilon(S) \ell (\mathsf f_S^\star (S))]  \label{eq:ittx0}
			\\
			&= L_S^\star(0) -  \EE[\varepsilon^\star(S) \ell (\mathsf f_S^\star (S))]
	\end{align}
where the optimization is over $\varepsilon \colon \mathbb Z^+ \mapsto [0, 1]$, and the optimal error profile $\varepsilon^\star(\cdot)$ that achieves \eqref{eq:ittx0} is given by \eqref{epsstar}.

An immediate consequence is that in the region of large error probability $\epsilon > 1 - P_S (1)$, $M =1$,   
all outcomes are mapped to $\varnothing$, and therefore, $L_{S, \texttt{det}}^\star (\epsilon) = 0$. At the other extreme, if $\epsilon = 0$, then $M = |\mathcal M|$ and \cite{verdu2014lossless}
\begin{align}\label{booster}
L_S^\star (0) = \mathbb{E} [ \ell (\mathsf f_S^\star (S)) ] = \sum_{i=1}^\infty \mathbb{P} [ S \geq 2^i ]
\end{align}

Denote by $L_{S, \texttt{det}}(\epsilon)$ the minimum average length comparable with error probability $\epsilon$ if randomized codes are not allowed. It satisfies the bounds
\begin{align}
L^\star_S(\epsilon) &\leq L_{S, \texttt{det}}(\epsilon)\\
&\leq  L^\star_S(\epsilon) +  \phi(\min\left\{\epsilon, e^{-1} \right\}) \label{eq:Lstardeterm}, 
\end{align}
where 
\begin{equation}
 \phi(x) \triangleq x \log_2 \frac 1 x \label{eq:phi}.
\end{equation}
Note that $0 \leq \phi(x)  \leq e^{-1} \log_2 e \approx 0.531$~bits on $x \in [0, 1]$, where the maximum is achieved at $x = e^{-1}$. 

To show \eqref{eq:Lstardeterm}, observe that the optimal encoder needs to randomize at most one element of $\mathcal M$. Indeed, let $m_0 \in \mathcal M$ be the minimum of $m_0$ satisfying
\begin{equation}
 \Prob{S > m_0  | \lfloor \log_2 S \rfloor = \eta } \leq \alpha \label{eq:m0}
\end{equation}
and map all $\{m > m_0 \colon \lfloor \log_2 m \rfloor = \eta \}$ to $\varnothing$, all $\{m < m_0 \colon \lfloor \log_2 m \rfloor = \eta \}$ to $\mathsf f^\star_S(m)$, and map $m_0$ to $\varnothing$ with probability 
\begin{equation}
 \alpha^- \triangleq \left(\alpha - \Prob{S > m_0  | \lfloor \log_2 S \rfloor = \eta } \right) \frac{\Prob{\lfloor \log_2 S \rfloor = \eta}}{P_S(m_0)},
\end{equation}
and to $\mathsf f^\star_S(m_0)$ otherwise. Clearly this construction achieves both \eqref{ladygaga} and \eqref{eq:lstar}. Using \eqref{eq:lmb0}, it follows that
\begin{align}
  L_{S, \texttt{det}}^\star(\epsilon) &=  L_{S}^\star(\epsilon) + \alpha^- P_S(m_0) \ell(\mathsf f^\star_S(m_0))\\
&\leq  L_{S}^\star(\epsilon) + \alpha^- P_S(m_0) \log_2 \frac 1 {P_S(m_0)}
\end{align}
To obtain \eqref{eq:Lstardeterm}, notice that $\alpha^- P_S(m_0) \leq \epsilon$, and if $P_S(m_0) > \epsilon$ we bound 
\begin{equation}
 \alpha^- P_S(m_0) \log_2 \frac 1 {P_S(m_0)} \leq \epsilon \log_2 \frac 1 \epsilon.
\end{equation}
Otherwise, since the function $\phi(p)$ is monotonically increasing on $p \leq e^{-1}$ and decreasing on $p > e^{-1}$, maximizing it over $[0, \epsilon]$ we obtain \eqref{eq:Lstardeterm}. 

Variants of the variational characterization \eqref{eq:ittx0} will be important throughout the paper. In general, for $X \in \mathbb R$
 \begin{equation}
 \E{ \left \langle X \right \rangle_\epsilon } = \min_{\varepsilon(\cdot): \EE[\varepsilon(X)] \le \epsilon} \E{(1 - \varepsilon(X)) X} \label{eq:Xepsvar}
\end{equation}
where the optimization is over $\varepsilon \colon \mathbb R \mapsto [0, 1]$.

\subsection{Erokhin's function}
As made evident in \eqref{eq:approx}, Erokhin's function \cite{erokhin1958epsilon} plays an important role in characterizing the nonasymptotic limit of variable-length lossless data compression allowing nonzero error probability. In this subsection, we point out some of its properties.

Erokhin's function is defined in \eqref{eq:heps_def2}, but in fact, the constraint in \eqref{eq:heps_def2} is achieved with equality: 
\begin{equation}\label{eq:itt5a}
		\mathbb H(S, \epsilon) = \min_{ \substack{ P_{Z | S } \colon \\ \Prob{S \neq Z} = \epsilon}} I(S; Z) 
\end{equation}
Indeed, given $\PP[S\neq Z]\le\epsilon$ we may define $Z'$ such that $S\to Z \to Z'$ and $\PP[S\neq Z']=\epsilon$ (for example, by
probabilistically mapping non-zero values of $Z$ to $Z'=0$).

Furthermore, Erokhin's function can be parametrically represented as follows \cite{erokhin1958epsilon}. 
	\begin{align}
			\mathbb{H} (S,\epsilon) 			&=  \sum_{m=1}^M P_S (m ) \log_2 \frac{1}{P_S (m) }
			- (1 - \epsilon ) \log_2\frac{1}{1-\epsilon} 
			\notag\\
			&- (M-1 ) \eta \log_2 \frac{1}{\eta}
\label{jamon} 
\end{align}
with the integer $M$ and $\eta >0 $ determined by $\epsilon$ through 
		\begin{align}\label{clocks}
			\sum_{m =1}^M P_S (m) = 1- \epsilon + (M-1) \eta
\end{align}	
In particular, $\mathbb{H} (S, 0)  = H(S)$, and if $S$ is equiprobable on an alphabet of $M$ letters, then 
	\begin{align}
			\mathbb{H} (S,\epsilon) =
\log_2 M - \epsilon \log_2 (M-1) -
 h(\epsilon)\,.
 \label{eq:heps_equi}
\end{align}

As the following result shows, Erokhin's function is bounded in terms of the expectation of the $\epsilon$-cutoff of information, $\left \langle \imath_S(S) \right \rangle_\epsilon$, which is easier to compute and analyze than the exact parametric solution in \eqref{jamon}.

\begin{thm} [Bounds to $\mathbb{H} (S, \epsilon)$] If $0 \le \epsilon < 1- P_S(1)$, Erokhin's function satisfies 
\begin{align} 
&~ 
\EE[\left \langle \imath_S(S) \right \rangle_\epsilon ] - \epsilon \log_2( L_S^\star(0)+\epsilon) -2\,h(\epsilon) - \epsilon\log_2{e\over \epsilon} 
\notag \\
\le&~
 \mathbb H(S, \epsilon)\label{eq:lmb5}\\
\le &~ \EE[\left \langle \imath_S(S) \right \rangle_\epsilon ] \label{eq:lmb6}
\end{align}
	
If $\epsilon \geq 1- P_S(1)$, then $\mathbb H(S, \epsilon)  = 0$.
\label{thm:Erokhinepsinfo}		
\end{thm}

\begin{proof}
 The bound in~\eqref{eq:lmb5} follows from~\eqref{eq:itt4} and~\eqref{eq:lmb3} below. Showing~\eqref{eq:lmb6} involves defining a
suboptimal choice (in~\eqref{eq:heps_def2}) of 
\begin{align}
Z = 
\begin{cases}
 S &  \left \langle \imath_S(S) \right \rangle_\epsilon >0 \\
 \bar S  &  \left \langle \imath_S(S) \right \rangle_\epsilon  = 0 
\end{cases}
\end{align} 
where $P_{S \bar S} = P_S P_S$, and noting that $I(S; Z) \leq D(P_{Z|S} \| P_{S} | P_S) = \EE[\left \langle \imath_S(S) \right \rangle_\epsilon ]$, where $D(\cdot \| \cdot | \cdot)$ denotes conditional relative entropy. 

\end{proof}

Figure~\ref{fig:I} plots the bounds to $\mathbb H(S^k, \epsilon)$ in Theorem \ref{thm:Erokhinepsinfo} for biased coin flips.

 \begin{figure}[htbp]
    \epsfig{file=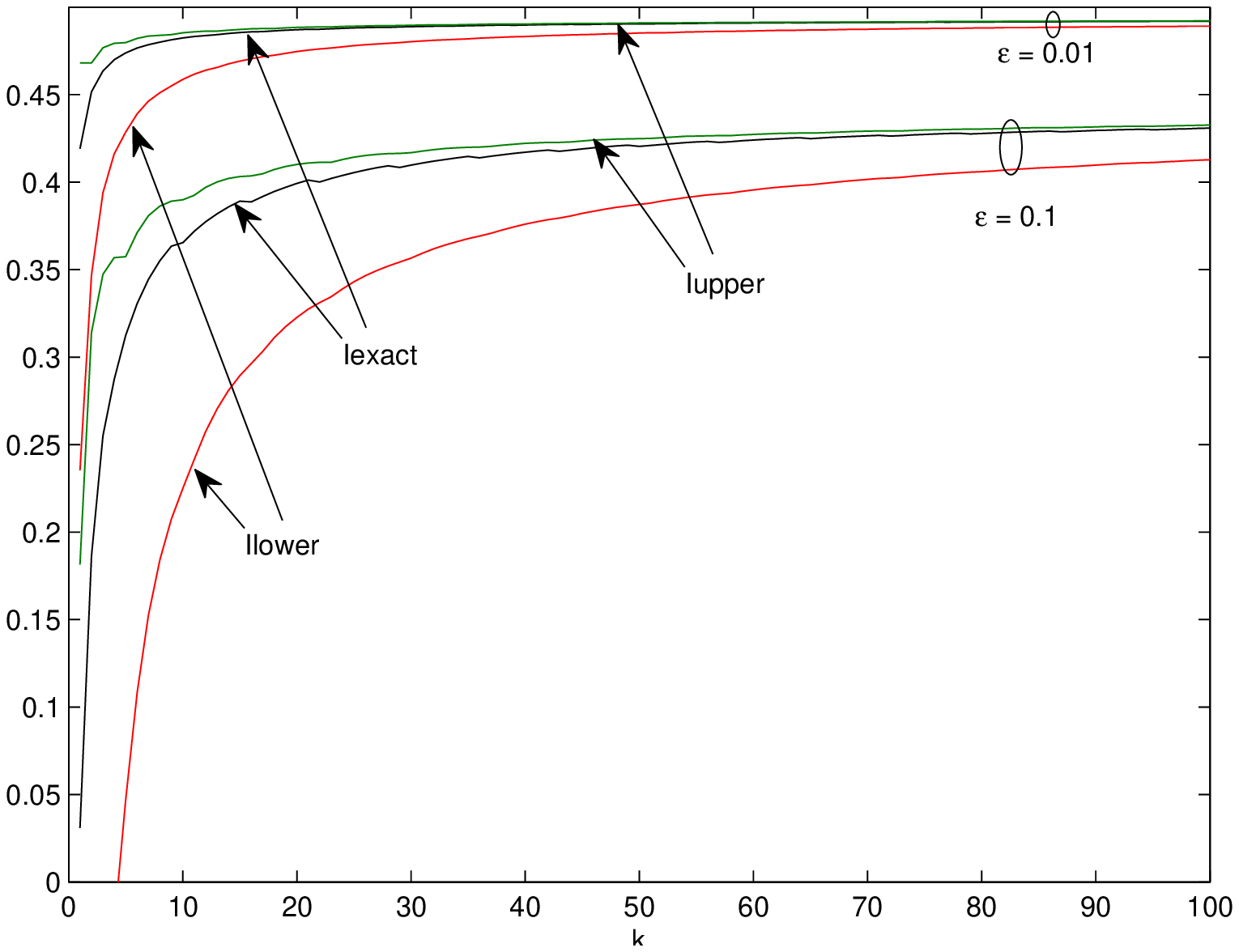,width=1\linewidth}
\caption{Bounds to Erokhin's function for a memoryless binary source with bias $p = 0.11$. }
\label{fig:I}
\end{figure}

\subsection{Non-asymptotic bounds}

Expression~\eqref{eq:lstar} is not always convenient to work with. The next result tightly bounds $L^\star(\epsilon)$ in terms of the $\epsilon$-cutoff of information, $\left \langle \imath_S(S) \right \rangle_\epsilon$, a random variable which is easier to deal with.

\begin{thm}[Bounds to $L_S^\star(\epsilon)$] If $0 \le \epsilon < 1- P_S(1)$, then the minimum achievable average length satisfies
 \begin{align}
\EE[ \left \langle \imath_S(S) \right \rangle_\epsilon ] +L_S^\star(0)-H(S) 
&\le L_S^\star(\epsilon) \label{eq:lmb3}\\
	&\le \EE[ \left \langle \imath_S(S) \right \rangle_\epsilon ] \label{eq:lmb4}
 \end{align}					    
If $\epsilon \geq 1- P_S(1)$, then $L_S^\star(\epsilon)  = 0$.
\label{thm:epsinfo}		
\end{thm}

\begin{proof}
Due to~\eqref{eq:Xepsvar}, we have the variational characterization:
\begin{align} \EE[ \left \langle \imath_S(S) \right \rangle_\epsilon] &= 
		H(S) - \max_{\varepsilon(\cdot): \EE[\varepsilon(S)] \le \epsilon} \EE[\varepsilon(S) \imath_S(S)] 
		\label{eq:lmb2}
\end{align}		
where $\varepsilon(\cdot)$ takes values in $[0,1]$. We obtain \eqref{eq:lmb3}--\eqref{eq:lmb4} comparing~\eqref{eq:ittx0} and~\eqref{eq:lmb2} via~\eqref{eq:lmb0}. 
\end{proof}

\begin{example}
If $S$ is equiprobable on an alphabet of cardinality $M$, then 
\begin{equation}
  \left \langle \imath_S(S) \right \rangle_\epsilon = 
\begin{cases}
\log_2 M \text{ w. p. } 1 - \epsilon\\
0 \text{ otherwise } 
\end{cases}
\end{equation}
\end{example}

The next result, in which the role of entropy is taken over by Erokhin's function, generalizes the bounds in \eqref{eq:alon2} and \eqref{eq:wyner} to $\epsilon > 0$.

\begin{thm}[Relation between $L_S^\star(\epsilon)$ and $\mathbb{H} (S, \epsilon)$] \label{thm:NOA}
If $0 \le \epsilon < 1- P_S(1)$, then the minimum achievable average length satisfies
	\begin{align} 
	&~ 
	\mathbb{H} (S, \epsilon) - \log_2 (\mathbb{H} (S, \epsilon) + 1) - \log_2 e 
	\notag \\
	\le&~ L_S^\star(\epsilon) \label{eq:itt0}\\
		\le&~ \mathbb{H} (S, \epsilon) + \epsilon \log_2 (H(S) + \epsilon) + \epsilon \log_2 {e\over\epsilon} + 2\,h(\epsilon) \label{eq:itt1} 
	\end{align}
	where $\mathbb{H} (S, \epsilon)$ is defined in \eqref{eq:heps_def2}, and the binary entropy function is denoted as $h(x) = {x\log_2 {1\over x}} + (1-x) \log_2 {1\over 1-x}$.

\end{thm}

Note that we recover \eqref{eq:alon2} and \eqref{eq:wyner} by particularizing Theorem \ref{thm:NOA} to $\epsilon = 0$. 
\par

\begin{proof}
 We first show the converse bound \eqref{eq:itt0}. The entropy of the output string $W\in \{ 0, 1 \}^\star$ of an arbitrary compressor $S \to W \to \hat S$ with $\Prob{S \neq \hat S} \leq \epsilon$ satisfies
\begin{align}\label{sw}
H (W ) \geq I ( S ; W ) = I ( S ; \hat{S} ) \geq \mathbb{H} (S,\epsilon) 
\end{align}
where the rightmost inequality holds in view of \eqref{eq:heps_def2}. 
Noting that the identity mapping $W \mapsto W \mapsto W$ is a lossless variable-length code, we lower-bound its average length as
\begin{align} \label{dishwalk}
H(W) - \log_2 (H(W) + 1 ) - \log_2 e &\leq L_W^\star ( 0 ) \\
&\leq \mathbb{E} [ \ell ( W ) ] \label{dishrun}
\end{align}
where \eqref{dishwalk} follows from  \eqref{eq:alon2}. The function of $H(W)$ in the left side of \eqref{dishwalk}
is monotonically increasing if $H(W) > \log_2 \frac{e}{2} = 0.44$~bits  and it is positive if  $H(W) > 3.66$~bits.
Therefore, it is safe to further weaken the bound in \eqref{dishwalk} by invoking \eqref{sw}. This concludes the proof of \eqref{eq:itt0}.
By applying \cite[Theorem 1]{szpankowski2011minimum} to $W$, we can get a sharper lower bound (which is always positive) 
\begin{align}
\psi^{-1} ( \mathbb{H} (S,\epsilon) ) \leq L_S^\star ( \epsilon ) \label{palmdrive}
\end{align}
where $\psi^{-1}$ is the inverse of the monotonic function on the positive real line:
\begin{align}
\psi (x) = x + (1+x) \log_2 (1 +x) - x \log_2 x .
\end{align}

To show the achievability bound \eqref{eq:itt1}, fix $P_{Z|S}$ satisfying the constraint in~\eqref{eq:itt5a}. Denote for brevity
\begin{align}
\Lambda &\triangleq  \ell (\mathsf f_S^\star (S))   \\
E &\triangleq 1\{S\neq Z\}\\
 \varepsilon(i) &\eqdef \PP[S\neq Z | \Lambda=i]
\end{align}
We proceed to lower bound the mutual information between $S$ and $Z$:
	\begin{align}
			I(S; Z) &= I(S; Z, \Lambda) - I(S; \Lambda|Z) \\
	&= H(S) - H(\Lambda|Z) -  H(S|Z, \Lambda) 
		\label{eq:itt2}\\
	&=	H(S) - I(\Lambda; E|Z) - H(\Lambda|Z,E) - H(S|Z, \Lambda)  \label{eq:itv1}\\
	&\geq  L_S^\star(\epsilon) + H(S) - L^\star_S(0) - \epsilon \log_2 (L^\star_S(0) + \epsilon) 
	\notag\\
	&
	-
	 \epsilon \log_2 {e\over\epsilon} - 2\, h(\epsilon) \label{draegers}
\end{align}	
where \eqref{draegers} follows from $I(\Lambda; E|Z) \leq h(\epsilon)$ and the following chains  \eqref{sandhill}-\eqref{hearst} and \eqref{eq:itv2}-\eqref{eq:itv6}.
\begin{align}
H(S| Z,  \Lambda)  
&\leq 
\EE[\varepsilon(\Lambda) \Lambda  + h(\varepsilon(\Lambda))] \label{sandhill}
\\
&\leq
L_S^\star(0) - L_S^\star(\epsilon) + h (\epsilon) \label{hearst}
\end{align}
where \eqref{sandhill} is by Fano's inequality:
 conditioned on $\Lambda = i$, $S$ can have at most $2^i$ values, so
\begin{align}\label{152mosher}
H(S| Z,  \Lambda =i) \leq~& i \, \varepsilon ( i)  + h ( \varepsilon ( i))
\end{align}
and \eqref{hearst} follows from ~\eqref{eq:ittx0}, \eqref{eq:itt5a} and the concavity of $h(\cdot)$.

The third term in \eqref{eq:itv1} is upper bounded as follows. 
	\begin{align} H(\Lambda|Z,E)
		&=  \epsilon H(\Lambda|Z, E=1) 
				\label{eq:itv2}\\
				&\le \epsilon H(\Lambda|S\neq Z)\label{eq:itv3}\\
	&\le  \epsilon \left( \log_2(1+\EE[\Lambda|S\neq Z]) + \log_2 e\right)
				\label{eq:itv4}\\
	&\le  \epsilon \left(\log_2\left( 1+{ \EE[\Lambda]\over \epsilon}\right) + \log_2 e\right)
				\label{eq:itv5}\\
	&= \epsilon \log_2 {e\over \epsilon} + \epsilon(\log_2( L_S^\star(0) + \epsilon)\,,
				\label{eq:itv6}
\end{align}
where~\eqref{eq:itv2} follows since $H(\Lambda|Z,E=0)=0$,~\eqref{eq:itv3} is because conditioning decreases
entropy,~\eqref{eq:itv4} follows by maximizing entropy under the mean constraint (achieved by the geometric 
distribution),~\eqref{eq:itv5} follows by upper-bounding 
	$$ \PP[S\neq Z]\, \EE[\Lambda|S\neq Z] \le \EE[\Lambda] $$
	and~\eqref{eq:itv6} applies  \eqref{booster}.
	
Finally, since the right side of \eqref{draegers} does not depend on $Z$, we may minimize the left side over $P_{Z|S}$ satisfying the constraint in~\eqref{eq:itt5a} to obtain 
\begin{align}\label{eq:itt4}
	 L_S^\star(\epsilon) &\le \mathbb H(S,\epsilon) + L_S^\star(0) - H(S) +\epsilon \log_2(L_S^\star(0) + \epsilon)  
	 \notag\\
	 &
	 + 2\, h(\epsilon) + \epsilon\log_2 {e\over
	\epsilon}
\end{align}	
 which leads to \eqref{eq:itt1} via Wyner's bound \eqref{eq:wyner}. 
 
 \end{proof}

\begin{remark}
The following stronger version of  \eqref{eq:alon2} is shown in \cite[Lemma 3]{alon1994lower}:
\begin{equation}
H(S)  \leq L_S^\star(0) + \log_2 (L_S^\star(0)  + 1)  + \log_2 e  \label{eq:alon2a}
\end{equation}
which, via the same reasoning as in \eqref{sw}--\eqref{dishrun}, leads to the following strengthening of \eqref{eq:itt0}:
\begin{equation}
	\mathbb{H} (S, \epsilon) 
	\le L_S^\star(\epsilon) + \log_2 (L_S^\star(\epsilon)+ 1) + \log_2 e 
\end{equation}
\label{rem:alon}
\end{remark}

Together, Theorems \ref{thm:Erokhinepsinfo}, \ref{thm:epsinfo}, and \ref{thm:NOA} imply that as long as the quantities $L_{S}^\star(\epsilon)$, $\mathbb H(S, \epsilon)$ and $\E{\left \langle \imath_S(S) \right \rangle_\epsilon}$ are not too small, they are close to each other.

In principle, it may seem surprising that $L_S^\star (\epsilon ) $ is connected to  $\mathbb{H} (S, \epsilon) $ in the way dictated by Theorem \ref{thm:NOA},
which implies that whenever the unnormalized quantity $\mathbb{H} (S, \epsilon) $ is large it must be close to the minimum average length. After all, the objectives of minimizing the 
input/output dependence and minimizing the description length of $\hat{S}$ appear to be disparate, and in fact $\eqref{festar}$ and the conditional distribution 
achieving \eqref{eq:heps_def2} are quite different: although in both cases $S$ and its approximation coincide on the most likely outcomes, the number of retained
outcomes is different, and to lessen dependence, errors in the optimizing conditional in \eqref{eq:heps_def2} do not favor $m=1$ or any particular outcome of $S$.


\subsection{Asymptotics for memoryless sources}

\begin{thm} \label{thm:Rlossless} 
 Assume that:
\begin{itemize}
\item  $P_{S^k} = P_{\mathsf S} \times \ldots \times P_{\mathsf S}$.
\item The third absolute moment of $\imath_{\mathsf S}(\mathsf S)$ is finite.
\end{itemize}
For any $0 \leq \epsilon \leq 1$ and $k\to\infty$ we have
\begin{small}
\begin{equation}
\left.\begin{aligned}
        & L_{S^k}^\star(\epsilon)\\
        & \mathbb{H} (S^k, \epsilon)\\
        &\E{ \left \langle \imath_{S^k}(S^k) \right \rangle_\epsilon}
       \end{aligned}
 \right\}
 = (1 - \epsilon) k H(\mathsf S) -   \sqrt{\frac{k V(\mathsf S)}{2 \pi} } e^{- \frac { (\Qinv{\epsilon})^2} 2 }  + \theta(k)  \label{Reps2order} 
\end{equation}
\end{small}
where the remainder term satisfies
\begin{equation}
-  \log_2 k + \bigo{\log_2 \log_2 k}  \leq \theta(k) \leq \bigo{1} \label{eq:remainder}
\end{equation}\end{thm}
\apxonly{Note: for $H(S^k, \epsilon)$ lower bound can be improved to $-\epsilon \log_2 k + O(1)$.}
\begin{proof}
If the source is memoryless, the information in $S^k$ is a sum of i.i.d. random variables as indicated in \eqref{eq:infosum}, and	Theorem \ref{thm:Rlossless} follows by applying Lemma \ref{lemma:EXeps} below to
	the bounds in Theorem~\ref{thm:epsinfo}.
\end{proof}

\begin{lemma}
Let $X_1, X_2, \ldots$ be a sequence of independent random variables with a common distribution $P_{\mathsf X}$ and a finite third absolute moment. Then for any $0 \leq \epsilon \leq 1$ and $k\to\infty$ we have
\begin{small}
 \begin{equation}
\E{ \left \langle \sum_{i = 1}^k X_i \right \rangle_{\!\!\epsilon\,} } = (1 - \epsilon) k \E{\mathsf X} -   \sqrt{\frac{k \Var{\mathsf X}}{2 \pi} } e^{- \frac { (\Qinv{\epsilon})^2} 2 } + \bigo{1} \label{EXeps}
\end{equation}
\end{small}
\label{lemma:EXeps}
\end{lemma}
\begin{proof}
Appendix \ref{appx:EXeps}.
\end{proof}

\begin{remark}
Applying \eqref{SVformula}  to \eqref{eq:lmb3}, for finite alphabet sources the lower bound on $L_{S^k}^\star(\epsilon)$ is improved to 
\begin{equation}
 \theta(k) \geq - \frac 1 2 \log_2 k + \bigo{1}
\end{equation}
For $\mathbb{H} (S^k, \epsilon)$, the lower bound is in fact $\theta(k) \geq -\epsilon \log_2 k + \bigo{1}$,  while for $\E{ \left \langle \imath_{S^k}(S^k) \right \rangle_\epsilon}$, $\theta(k) = \bigo{1}$. 
\end{remark}

\begin{remark}
 If the source alphabet is finite, we can sketch an alternative proof of Theorem \ref{thm:Rlossless} using the method of types. 
 By concavity and symmetry, it is easy to see that the optimal coupling
that achieves $\mathbb{H} ( S^k, \epsilon )$ satisfies the following property: the error profile
\begin{align} 
\epsilon(s^k) \eqdef \PP[Z^k \neq S^k|S^k = s^k] 
\end{align}
is constant on each $k$-type (see~\cite[Chapter 2]{csiszar2011information} for types). Denote the type of $s^k$ as $\hat P_{s^k}$
and its size as $M(s^k)$. We then have the following chain:
\begin{align} I(S^k; Z^k) &= I(S^k, \hat P_{S^k}; Z^k)\\
	&= I(S^k; Z^k | \hat P_{S^k}) + O(\log k)\label{eq:htt1}\\
	&\ge \EE\left[ (1-\epsilon(S^k)) \log M(S^k) \right] + O(\log k) \label{eq:htt2}
\end{align}			
where~\eqref{eq:htt1} follows since there are only polynomially many types and~\eqref{eq:htt2} follows
from~\eqref{eq:heps_equi}. Next,~\eqref{eq:htt2} is to be minimized over all $\epsilon(S^k)$ satisfying
$\EE[\epsilon(S^k)] \le \epsilon$. The solution (of this linear optimization) is easy:
$\epsilon(s^k)$ is 1 for all types with $M(s^k)$ exceeding a certain threshold, and 0 otherwise. In other words, we get
\begin{equation}\label{eq:htt3a}
	\mathbb{H} (S^k, \epsilon) = (1-\epsilon) \EE[\log M(S^k)|M(S^k) \le \gamma] + O(\log k)\,,
\end{equation}
where $\gamma$ is chosen so that $\PP[M(S^k) > \gamma]=\epsilon$. Using the relation between type size and its
entropy, we have
\begin{equation}\label{eq:htt3}
	\log	M(s^k) = k H(\hat P_{s^k}) + O(\log k) 
\end{equation}	
and from the central-limit theorem, cf.~\cite{yushkevich1953limit,strassen1962asymptotische}, we get
\begin{equation}\label{eq:htt4}
		H(\hat P_{S^k}) \stackrel{d}{=}  H(\mathsf S) + \sqrt{V(\mathsf S)\over k} U + O\left(\frac{\log k} k\right)\qquad U \sim\matn(0,1)\,.
\end{equation}
Thus, putting together~\eqref{eq:htt3a},~\eqref{eq:htt3},~\eqref{eq:htt4} and after some algebra~\eqref{Reps2order} follows.
\end{remark}

\subsection{Discussion}\label{sec:discussion}

Theorem \ref{thm:Rlossless} exhibits an unusual phenomenon in which the dispersion term improves the achievable average rate. As illustrated in Fig. \ref{fig:dispersion}, a nonzero error probability $\epsilon$ decreases the average achievable rate as the source outcomes falling into the shaded area are assigned length 0. The total reduction in average length is composed of the reduction in asymptotically achievable average length due to nonzero $\epsilon$ and the reduction due to finite blocklength. The asymptotic average length is reduced because the center of probabilistic mass Fig. \ref{fig:dispersion} shifts to the left when the $\epsilon$-tail of the distribution is chopped off. Moreover, for a fixed $\epsilon$ the wider the distribution the bigger is this shift, thus shorter blocklengths and larger dispersions help to achieve a lower average rate.

\begin{figure}[htbp]
\begin{center}
\epsfig{file=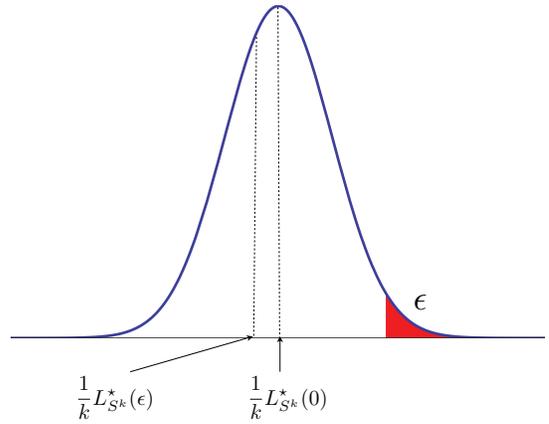,width=.8\linewidth} 
\end{center}    
\caption{The benefit of nonzero $\epsilon$ and dispersion. The bell-shaped curve depicts an idealized form of  the pmf of $\frac 1 k \ell \left( f^\star(S^k) \right)$.}
\label{fig:dispersion}
\end{figure}

For a source of biased coin flips, Fig.~\ref{fig:rate} depicts the exact average rate of the optimal code as well as the approximation in \eqref{Reps2order}. Both curves are monotonically increasing in $k$.

\begin{figure}[htbp]
    \epsfig{file=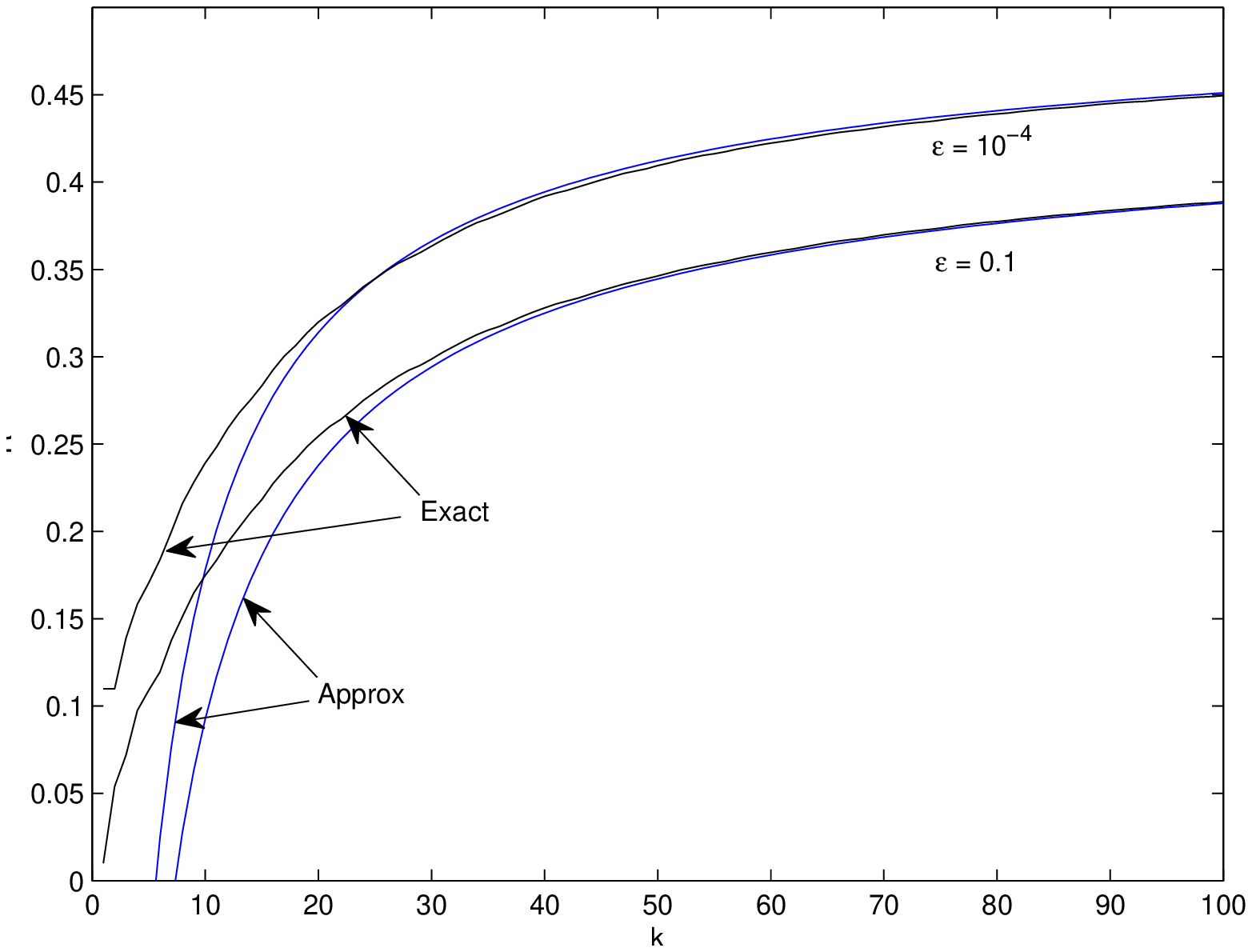,width=1\linewidth}
\caption{Average rate achievable for variable-rate  almost lossless encoding of a memoryless binary source with bias $p = 0.11$ and two values of $\epsilon$. 
For $\epsilon < 10^{-4}$, the resulting curves are almost indistinguishable from the $\epsilon = 10^{-4}$ curve.
}
\label{fig:rate}
\end{figure}

The dispersion term in \eqref{Reps2order} vanishes quickly with $\epsilon$. More precisely, as $\epsilon \to 0$, we have (Appendix \ref{appx:disp0})
\begin{equation}
\frac 1 {\sqrt{2 \pi}}  e^{- \frac {\left( \Qinv{\epsilon} \right)^2} 2} 
= 
 \epsilon \sqrt {2 \log_2 \frac 1 \epsilon}  + \smallo{\epsilon} \label{disp0}
\end{equation} 

Therefore, a refined analysis of higher order terms in the expansion \eqref{Reps2order} is desirable in order to obtain an approximation which is accurate even at short blocklengths. Inspired by \cite{szpankowski2008antiredundancy}, in Fig.~\ref{fig:rate} (devoted to independent coin flips with bias $p$) we adopt the following value for the remainder in \eqref{Reps2order}: 
\begin{align}
\theta(k) = (1 - \epsilon) \bigg (& \frac{\log_2 k}{2}  - \frac 1 2 \log_2 (4 e^3 \pi) + \frac p {1 - 2 p} 
\\
&+
 \log_2 \frac 1 {1 - 2 p} + \frac 1 {2 (1 - 2 p) } \log_2 \frac {1 - p} p  \bigg ),  \notag
\end{align}
which proves to yield a remarkably good approximation, accurate for blocklengths as short as $20$.

\apxonly{

TODO: Can we prove
\begin{equation}
\frac 1 {k + 1} L_{S^{k + 1}}(\epsilon) \geq  \frac 1  k L_{S^{k}}(\epsilon)
\end{equation}

TODO: prove the following refinement of \eqref{Reps2order}: 
\begin{equation}
 L_{S^k}(\epsilon) = (1 - \epsilon) k H(\mathsf S) -   \sqrt{\frac{k V(\mathsf S)}{2 \pi} } e^{- \frac { (\Qinv{\epsilon})^2} 2 }  - (1 - \epsilon) \frac{\log_2 k}{2} + \bigo{1} 
\end{equation}
This would generalize the result in \cite{szpankowski2011minimum} on $R(k, 0)$.

TODO: prove the following refinement of \eqref{Reps2order} for the binary source with bias $p$: 
\begin{align}
 L_{S^k}(\epsilon) &= (1 - \epsilon) k H(\mathsf S) -   \sqrt{\frac{k V(\mathsf S)}{2 \pi} } e^{- \frac { (\Qinv{\epsilon})^2} 2 }  - (1 - \epsilon) \frac{\log_2 k}{2} \\
&+ (1 -\epsilon) \left( - \frac 1 2 \log_2 (4 e^3 \pi) + \frac p {1 - 2 p} + \log_2 \frac 1 {1 - 2 p} + \frac 1 {2 (1 - 2 p) } \log_2 \frac {1 - p} p  \right) +  \smallo{1} \label{Reps2orderBinary} 
\end{align}
This would generalize the result in \cite{szpankowski2008antiredundancy}, applicable for $\epsilon = 0$.


}

\begin{figure}[htbp]
\psfrag{k}{\tiny{$k$}}
    \epsfig{file=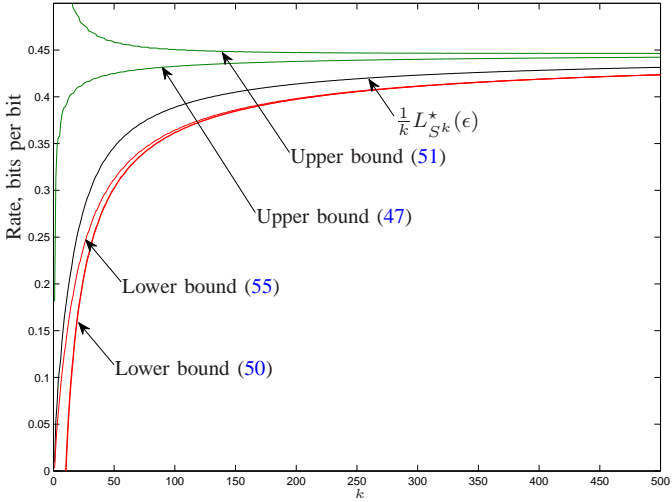,width=1\linewidth}
\caption{Bounds to the average rate achievable for variable-rate  almost lossless encoding of a memoryless binary source with bias $p = 0.11$ and $\epsilon = 0.1$. The lower bound in \eqref{eq:itt0} is virtually indistinguishable from a weakening of \eqref{eq:lmb3} using \eqref{eq:alon2}. }
\label{fig:rate}
\end{figure}

\section{Lossy variable-length compression}
\label{sec:lossy}
\subsection{The setup}

In the basic setup of lossy compression, we are given a source alphabet $\mathcal M$,  a reproduction alphabet $\widehat {\mathcal M}$, a {\it distortion measure} $\mathsf d \colon \mathcal M \times \widehat{ \mathcal M} \mapsto [0, + \infty]$ to assess the fidelity of reproduction, and a probability distribution of the object $S$ to be compressed. 

\begin{defn}[$(L, d, \epsilon)$ code]
 A variable-length $(L, d, \epsilon)$ lossy code for $\{S, \mathsf d\}$ is a pair of random transformations $P_{W|S} \colon \mathcal M \mapsto \left\{0, 1\right\}^\star$ and $P_{Z|W} \colon \left\{0, 1\right\}^\star \mapsto \widehat {\mathcal M}$ such that
 \begin{align}
\Prob{\mathsf d \left( S, Z \right) > d} &\leq \epsilon\\
\E{ \ell( W ) } &\leq L
\end{align}
\end{defn}

The goal of this section is to characterize the minimum achievable average length compatible with the given tolerable error $\epsilon$: 
\begin{equation}
L_S^\star(d, \epsilon) \triangleq \left\{ \min ~ L \colon \exists \text{ an $(L, d, \epsilon)$ code} \right\} 
\end{equation}

Section \ref{sec:optlossy} discusses the properties of the optimal code. Section~\ref{sec:rd} reviews some background facts from rate-distortion theory. Section \ref{sec:1shotlossy} presents single-shot results, and Section \ref{sec:2orderlossy} focuses on the asymptotics.

\subsection{Optimal code}
\label{sec:optlossy}
Unlike the lossless setup in Section \ref{sec:lossless}, the optimal encoding and decoding mappings do not admit, in general, explicit descriptions. We can however point out several  properties of the optimal code. 

We first focus on the case $\epsilon = 0$. The optimal $(d, 0)$ code satisfies the following properties. 
\begin{enumerate}
\item The optimal encoder $\mathsf f^\star$ and decoder $\mathsf g^\star$ are deterministic mappings. \label{opt:1}
\item The output $W^\star = \mathsf f^\star(S)$ of the optimal encoder satisfies $P_{W^\star}(\varnothing) \geq P_{W^\star}(0) \geq P_{W^\star}(1) \geq P_{W^\star}(00) \geq \ldots$  \label{opt:2}
\item 
For each $w \in \{0, 1\}^\star$
\begin{align}
{\mathsf f^{\star}}^{ -1}(w) = B_{\mathsf g^\star(w)} \backslash \cup_{v \prec w} B_{\mathsf g^\star(v)}
\end{align}
where $\prec$ is lexicographic ordering, and
\begin{equation}
B_z \triangleq \left\{ s \colon \mathsf d(s, z) \leq d\right\} 
\end{equation}
\label{opt:3}
\end{enumerate}

Let $z_1, z_2, \ldots$ be a $d$-covering  of $\mathcal M$. First, we will show that the foregoing claims hold for decoders whose image is constrained to the given $d$-covering $z_1, z_2, \ldots$. Then, we will conclude that since the claims hold for all $d$-coverings, they hold for the one that results in the minimum average length as well.  

To show \ref{opt:1}), let $(P_{W|S}, P_{Z|W})$ be a $(d, 0)$ code.  The optimal encoder is deterministic because if there exist $s \in \mathcal M$ and $w \prec v \in \{0, 1\}^\star$ such that $P_{W|S = s}(w) > 0$ and $P_{W|S = s}(v) > 0$ we may decrease the average length by setting $P_{W|S = s}(w) = 1$. The optimal decoder is deterministic because if for some $w \in \{0, 1\}^\star$ there exist $z^\prime, z^{\prime\prime} \in \{z_1, z_2, \ldots\}$ such that $P_{Z|W = w}(z^\prime) > 0$ and  $P_{Z|W = w}(z^{\prime\prime}) > 0$, then nothing changes by setting $P_{Z|W = w}(z^\prime) = 1$. 

To show \ref{opt:2}), observe that if there exist $w \prec v \in \{0, 1\}^\star$ such that $P_{W}(w) < P_W(v)$, then the average length is shortened by swapping $w$ and $v$. 

To show \ref{opt:3}), notice that the average length decreases as $P_W(\varnothing)$ increases, and the latter is maximized by setting $\mathsf f^{-1}(\varnothing) = B_{\mathsf g(\varnothing)}$. Further, $P_W(0)$ is maximized without affecting $P_W(\varnothing)$ by setting $\mathsf f^{-1}(0) = B_{\mathsf g(0)} \backslash B_{\mathsf g(\varnothing)}$ and so forth.

We now consider the case $\epsilon > 0$. The optimal $(d, \epsilon)$ code satisfies the following properties. 
\begin{enumerate}
\item The optimal decoder $\mathsf g^\star$ is deterministic, and the optimal encoder $P_{W^\star|S}$ satisfies $P_{W^\star|S = s}(w) = 1 - P_{W^\star|S = s}(\varnothing)$ for all $s \in \mathcal M$ and all $w \in \{0, 1\}^\star \backslash \varnothing$. 
  \label{opt:1a}
\item The output of the optimal encoder satisfies $P_{W^\star}(\varnothing) \geq P_{W^\star}(0) \geq P_{W^\star}(1) \geq P_{W^\star}(00) \geq \ldots$  \label{opt:2a}
\item There exist $\eta \in \mathbb R^+$ such that $\Prob{\ell(W^\star) > \eta} = 0$ and $0 \leq \alpha < 1$ such that
for each $w \in \{0, 1\}^\star \backslash \varnothing$
\begin{align}
&~P_{W^\star|S = s}(w) \\
&~= 
\begin{cases}
 1, ~ s \in B_{\mathsf g^\star(w)} \backslash \cup_{v \prec w} B_{\mathsf g^\star(v)} ~\&~ \ell(w) < \eta\\
 1 - \alpha, ~ s \in B_{\mathsf g^\star(w)} \backslash \cup_{v \prec w} B_{\mathsf g^\star(v)} ~\&~ \ell(w) = \eta \notag\\
\end{cases}
\end{align}
and 
\begin{align}
P_{W^\star|S = s}(\varnothing) = 
\begin{cases}
 1, ~ s \notin  \cup_{w} B_{\mathsf g^\star(w)} \\
 \alpha, ~ s \in  \cup_{w} B_{\mathsf g^\star(w)} ~\&~ \ell(w) = \eta\\
\end{cases}
\end{align}
\label{opt:3a}
\end{enumerate}
Property \ref{opt:3a}) implies in particular that $\ell(\mathsf f^\star(s)) = 0$ as long as $\mathsf d(s, \mathsf g^\star(\mathsf f^\star(s))) > d$. 

We say that $\mathcal F \subseteq \widehat {\mathcal M}$ is a $(d, \epsilon)$-covering  of $\mathcal M$ if $\Prob{ \min_{z \in \mathcal F} \mathsf d(S, z) > d} \leq \epsilon$. Note that a finite $(d, \epsilon)$-covering always exists as long as a $d$-covering exists \cite{posner1967epsilonentropy}: indeed, given a $d$-covering $z_1, z_2, \ldots$, let $M$ satisfy $\sum_{m > M} \Prob{S \in B_{z_m} \backslash \cup_{i < m } B_{z_i}} \leq \epsilon$ and just drop all $z_m \colon m > \eta$ to obtain a finite $(d, \epsilon)$-covering. 
Let $z_1, z_2, \ldots, z_M$ be a $(d, \epsilon)$-covering  of $\mathcal M$. Observing that an infinite  $(d, \epsilon)$-covering $z_1, z_2, \ldots$ can only result in a longer average length, we will first show that the foregoing claims hold for decoders whose image is constrained to a given $(d, \epsilon)$-covering $z_1, z_2, \ldots, z_M$. Then, we will conclude that since the claims hold for all finite $(d, \epsilon)$-coverings, they hold for the one that results in the minimum average length as well.

To show \ref{opt:1a}), notice that for a given encoder $P_{W|S}$, the optimal decoder is always
	deterministic. Indeed, if for some $w \in \{0, 1\}^\star$ there exist $z^\prime, z^{\prime\prime} \in \{z_1, z_2, \ldots, z_M\}$ such that $P_{Z|W = w}(z^\prime) > 0$,  $P_{Z|W = w}(z^{\prime\prime}) > 0$ and $P_{S|W = w}( B_{z^\prime} ) \geq P_{S|W = w}( B_{z^{\prime\prime}}  )$ then the excess distortion can only be reduced by setting $P_{Z|W = w}(z^\prime) = 1$, without affecting the average length. Denote that deterministic decoder by $\mathsf g$. As for the encoder, consider $w \in \{0,1\}^\star \setminus \varnothing$ and source realization $s$
	with $P_{W|S = s}(w)>0$. If $\mathsf d( s, \mathsf g(w)) > d$, the average length can be decreased, without increasing
	the excess distortion probability, by setting $P_{W|S = s}(w) = 0$ and adjusting $P_{W|S = s}(\varnothing) = 1$ accordingly. 
	This
	argument implies that the optimal encoder satisfies $P_{S| W = w}(B_{\mathsf g(w)}) = 1$ for each $w \neq \varnothing$. Now, if there exist $s$ and $w \prec v \in \{0, 1\}^\star \backslash \varnothing$ such that $P_{W|S = s}(w) > 0$ and $P_{W|S = s}(v) > 0$, we may decrease the average length with no impact on the probability of excess distortion by setting $P_{W|S = s}(w) = 1$.  
	
To show \ref{opt:2a}), notice that if there exist $w \prec v \in \{0, 1\}^\star \backslash \varnothing$ such that $P_{W}(w) < P_W(v)$, then the average length is shortened by swapping $w$ and $v$. If there exist $w \in \{0, 1\}^\star \backslash \varnothing$ with $P_W(w) > P_W(\varnothing)$ then the average length is shortened by swapping $w$ and $\varnothing$ and setting $P_{W|S = s}(w) = 0$ while adjusting $P_{W|S = s}(\varnothing) = 1$ accordingly for each $s \notin B_{g(w)}$. 

To show \ref{opt:3}),  we argue as in the case $\epsilon = 0$ that setting  
\begin{align}
 P_{W|S = s}(w) &= 1, ~ s \in B_{\mathsf g(w)} \backslash \cup_{v \prec w} B_{\mathsf g(v)}\\
  P_{W|S = s}(\varnothing) &= 1, ~ s \notin  \cup_{w} B_{\mathsf g(w)}
\end{align}
yields the minimum average length among all $(d, \epsilon^\prime)$ codes with codebook $z_1, z_2, \ldots$ satisfying \ref{opt:1a}) and \ref{opt:2a}) where $\epsilon^\prime \triangleq\Prob{ \min_{m} \mathsf d(S, z_m) > d}$. If $\epsilon^\prime = \epsilon$, there is nothing else to prove. If $\epsilon^\prime < \epsilon$, let $\eta \in \mathbb R^+$ and $0 < \alpha < 1$ solve
\begin{equation}
 \Prob{\ell( W) > \eta } + \alpha  \Prob{\ell( W) = \eta } = \epsilon - \epsilon^\prime
\end{equation}
 and observe that dropping all $w \colon \ell( w) > \eta$ reduces the average length while keeping the excess distortion probability below $\epsilon$. Now, letting $P_{W|S = s}(w) = 1 - \alpha$ for each $s \in B_{\mathsf g(w)} \backslash \cup_{v \prec w} B_{\mathsf g(v)}$ and each $w \colon \ell( w) = \eta$ and adjusting 
$P_{W|S = s}(\varnothing)$ accordingly further reduces the average length while making the excess distortion probability exactly $\epsilon$.

Property \ref{opt:3}) implies that randomization is not essential as almost the same average length can be achieved with deterministic encoding and decoding operations. Precisely,  denoting by $L^\star_{S, \texttt{det}}(d, \epsilon)$ the minimum average length achievable with deterministic codes, we have
\begin{align}
L^\star_S(d, \epsilon) &\leq L^\star_{S, \texttt{det}}(d, \epsilon)\\
&\leq  L^\star_{S}(d, \epsilon) + \phi(\min\{\epsilon, e^{-1}\}) \label{eq:Lstarlossydeterm}  
\end{align}
where \eqref{eq:Lstarlossydeterm} is obtained in the same way as \eqref{eq:Lstardeterm}, and $0 \leq \phi(\cdot) \leq 0.531$ is defined in \eqref{eq:phi}.

\subsection{A bit of rate-distortion theory}
\label{sec:rd}

The minimal mutual information function
\begin{equation}
 \mathbb R_S(d) \triangleq  \inf_{\substack{P_{Z|S}\colon\\ \E{\mathsf d(S, Z)} \leq d}} I(S; Z) \label{RR(d)}
\end{equation}
characterizes the minimum asymptotically achievable rate in both fixed-length compression under the average or excess distortion constraint and variable-length lossy compression under the almost sure distortion constraint \cite{shannon1959coding,kieffer1991strong}. 

We assume throughout that the following basic assumptions are met.
\begin{enumerate}[(A)]
 \item   $\mathbb R_S(d)$ is finite for some $d$, i.e. $ d_{\min} < \infty$, where
\begin{equation}
 d_{\min} \triangleq \inf \left\{ d\colon ~ \mathbb R_S(d) < \infty \right\} \label{dmin}
\end{equation}
\label{item:a}
\item  The distortion measure is such that there exists a finite set $E \subset \widehat {\mathcal M}$ such that
\begin{equation}
 \E{ \min_{z \in E} \mathsf d(S, z)} < \infty \label{dcsiszar}
\end{equation}
 \label{item:b}
\end{enumerate}

The following characterization of $\mathbb R_S(d)$ due to Csisz\'ar \cite{csiszar1974extremum} will be instrumental. 
\begin{thm} [{Characterization of $\mathbb R_S(d)$  \cite[Theorem 2.3]{csiszar1974extremum}}]
 \label{thm:csiszarg}
For each $d > d_{\min}$ it holds that
\begin{equation}
 \mathbb R_S(d) = \max_{J(s), ~ \lambda}\left\{ \E{ J(S)} - \lambda d\right\} \label{RR(d)csiszar}
\end{equation}
where the maximization is over $J(s)\geq 0$ and $\lambda\geq 0$ satisfying the constraint
\begin{equation}
\E{ \exp\left\{ J(S) - \lambda \mathsf d (S, z)\right\}} \leq 1 ~ \forall z \in \widehat{\mathcal M} \label{csiszarg}
\end{equation}
\end{thm}
Let $(J_S(s), \lambda_S)$ attain the maximum in the right side of \eqref{RR(d)csiszar}. If there exists a transition probability kernel $P_{Z^\star | S}$ that actually achieves the infimum in the right side of \eqref{RR(d)}, then \cite{csiszar1974extremum}
\begin{align}
J_S(s) &= \imath_{S; Z^\star} (s; z) + \lambda_S \mathsf d(s, z) \label{Jinfo}\\
&= -\log_2 \E{\exp \left( - \lambda_S \mathsf d(s, Z^\star)\right) } \label{JE}
\end{align}
where \eqref{Jinfo} holds for $P_{Z^\star}$-a.e. $z$, the expectation in \eqref{JE} is with respect to the unconditional distribution of $Z^\star$, and the usual information density is denoted by
\begin{equation}
 \imath_{S; Z} (s; z) \triangleq \log_2 \frac{d P_{Z|S = s}}{dP_Z} (z) 
\end{equation}
Note from \eqref{JE} that by the concavity of logarithm
\begin{equation}
0 \leq J_S(s) \leq \E{\mathsf d(s, Z^\star)} \label{Jstarbounds}
\end{equation}
The random variable that plays the key role in characterizing the nonasymptotic fundamental limit of lossy data compression is the $\mathsf d$-tilted information in $s \in \mathcal M$ \cite{kostina2011fixed}:
\begin{equation}
\jmath_{S}(s, d) \triangleq J_S(s) - \lambda_S d \label{dtilted}
\end{equation}
It follows from \eqref{RR(d)csiszar} that
\begin{equation}
 \mathbb R_{S}(d) = \E{\jmath_S(S, d)}  \label{RdEj}
\end{equation}
Much like information in $s \in \mathcal M$ which quantifies the number of bits necessary to represent $s$ losslessly, $\mathsf d$-tilted information in $s$ quantifies the number of bits necessary to represent $s$ within distortion $d$, in a sense that goes beyond average as in \eqref{RdEj} \cite{kontoyiannis2000pointwise,kostina2011fixed}. 
Particularizing \eqref{csiszarg}, we observe that the $\mathsf d$-tilted information satisfies
 \begin{equation}
\E{\exp(\jmath_S(S, d) + \lambda_S d - \lambda_S \mathsf d(S, z)) } \leq 1 \label{csiszarg1}
\end{equation}

Using Markov's inequality and \eqref{JE}, it is easy to see that the $\mathsf d$-tilted information is linked to the probability that $Z^\star$ falls within distortion $d$ from $s \in \mathcal M$:
\begin{equation}
 \jmath_{S}(s, d)  \leq \log_2 \frac 1 {P_{Z^\star}(B_d(s))} \label{dtiltedupper}
\end{equation}
 where
\begin{equation}
B_d(s) \triangleq  \left\{ z \in \widehat {\mathcal M} \colon \mathsf d(s, z) \leq d \right\} \label{dball}
\end{equation}
Moreover, under regularity conditions the reverse inequality in \eqref{dtiltedupper} can be closely approached~\cite[Proposition 3]{kontoyiannis2000pointwise}. 


\subsection{Nonasymptotic bounds}
\label{sec:1shotlossy}
We begin with a simple generalization of basic bounds \eqref{eq:alon2} and \eqref{eq:wyner} to an arbitrary distortion measure and nonzero $\epsilon$, in which the role of entropy is assumed by the $\left(\epsilon, \delta\right)$-entropy of the source $S$, defined as \cite{posner1967epsilonentropy}: 
\begin{equation}
H_{\epsilon, \delta}(S) \triangleq \min_{\substack{ \mathsf f \colon \mathcal M \mapsto \widehat {\mathcal M} \colon \\ \Prob{\mathsf d(S, \mathsf f (S) ) > \epsilon} \leq \delta}} H(\mathsf f (S)).  \label{eq:Hepsdelta}
\end{equation}

\begin{thm}[Bounds to $L_{S, \texttt{det}}^\star(d, \epsilon)$]
The minimal average length achievable with deterministic codes under an excess-distortion constraint satisfies
\begin{align}\label{eq:alon2lossyeps}
		H_{d, \epsilon}(S) - \log_2 (H_{d, \epsilon}(S)+1) - \log_2 e &\le L^\star_{S, \texttt{det}}(d, \epsilon) \\
		&\le H_{d, \epsilon}(S) \label{eq:wynerlossyeps}
\end{align}
\label{thm:alonlossy}
\end{thm}
\begin{proof}
The converse bound  in \eqref{eq:alon2lossyeps} follows by applying \eqref{eq:alon2} and minimizing over all possible output entropies. The achievability bound in \eqref{eq:wynerlossyeps} is implied by Wyner's bound \eqref{eq:wyner} recalling (Section \ref{sec:optlossy}) that the codewords of the optimal code are ordered in decreasing probabilities. 
\end{proof}

Note that $L^\star(d, \epsilon)$ is also bounded in terms of $H_{d, \epsilon}(S)$, in view of Theorem \ref{thm:alonlossy} and \eqref{eq:Lstarlossydeterm}.

Particularizing Theorem \ref{thm:alonlossy} to $\epsilon = 0$ and using  $L_S^\star(d, 0) = L_{S, \texttt{det}}^\star(d, 0)$ (as shown in Section \ref{sec:optlossy}), we see that the minimum average length of $d$-semifaithful codes is bounded by
\begin{align}\label{eq:alon2lossy}
		H_d(S) - \log_2 (H_d(S)+1) - \log_2 e &\le L_S^\star(d, 0) \\
		&\le H_d(S)\,, \label{eq:wynerlossy}
\end{align}
where  
$H_\epsilon(S)$ is the $\epsilon$-entropy of the source $S$ \cite{posner1967epsilonentropy}: 
\begin{equation}
H_\epsilon(S) \triangleq \min_{\substack{ \mathsf f \colon \mathcal M \mapsto \widehat {\mathcal M} \colon \\ \mathsf d(S, \mathsf f (S) ) \leq \epsilon \text{ a.s.}}} H(\mathsf f (S)),  \label{eq:Heps}
\end{equation}
which is bounded as follows: 
\begin{align}
 {\mathbb R}_{S}(d, 0) 
\leq&~ H_d(S) \\
\leq&~ \mathbb R_S(d,0) + \log_2 \mathbb ( \mathbb R_S(d, 0) + 1) + C,  \label{Alossy0}
\end{align}
where $C$ is a universal constant, and 
\eqref{Alossy0} holds whenever $\mathsf d$ is a metric by \cite[Theorem 2]{posner1971epsilon}.

Theorem \ref{thm:alonlossy} applies to the almost-lossless setting of Section \ref{sec:lossless}, in which case the $\left(\epsilon, \delta\right)$-entropy particularizes to $\epsilon = 0$ and Hamming distortion as
\begin{equation}
H_{0, \delta}(S) = \min_{\substack{ \mathsf f \colon \mathcal M \mapsto \widehat {\mathcal M} \colon \\ \Prob{S \neq \mathsf f (S) } \leq \delta}} H(\mathsf f (S)).  \label{eq:Hepsdeltalossless}
\end{equation}

The $\left(\epsilon, \delta\right)$-entropy is difficult to compute and analyze directly. We proceed to give bounds on $L^\star_S(d, \epsilon)$ and $H_{d, \epsilon}(S)$ that will essentially show that all the functions $L^\star_S(d, \epsilon)$, $H_{d, \epsilon}(S)$, $ {\mathbb R}_{S}(d, \epsilon)$ (defined in \eqref{RR(eps)}), are within $\bigo{\log_2 \mathbb R_{S}(d)}$ bits from the easy-to-analyze function $\EE[ \left \langle \jmath_S(S, d) \right \rangle_\epsilon]$. 
We will show that the same is true for the function
\begin{equation}
  {\mathbb R}_{S}^+(d, \epsilon) \triangleq \inf_{P_Z}\E{ \left \langle - \log_2 P_{Z}(B_d(S)) \right \rangle_\epsilon },  \label{RR(eps)+}
\end{equation}
where $B_d(s)$ is the distortion $d$-ball around $s$ (formally defined in \eqref{dball}) and the infimum is over all distributions on $\widehat {\mathcal M}$,

The next result provides nonasymptotic bounds to the minimum achievable average length when randomized encoding and decoding operations are allowed. 
\begin{thm}[Bounds to $L_S^\star(d, \epsilon)$]
The minimal average length achievable under an excess-distortion constraint satisfies
\begin{align}
 {\mathbb R}_{S}(d, \epsilon) - \log_2 \left( {\mathbb R}_{S}(d, \epsilon)   + 1 \right) - \log_2 e
\leq&~ L_S^\star(d, \epsilon) \label{Clossy}\\
\leq&~   {\mathbb R}_{S}^+(d, \epsilon)  \label{Alossy}
\end{align}
 where $\mathbb R_{S}(d, \epsilon)$ is the minimal information quantity defined in \eqref{RR(eps)}, and ${\mathbb R}_{S}^+(d, \epsilon)$ is defined in \eqref{RR(eps)+}. 
\label{thm:Elenlossy}
\end{thm}

\begin{proof}
The converse bound in \eqref{Clossy} is shown in the same way as \eqref{eq:alon2lossyeps}.  To show the achievability bound in \eqref{Alossy}, consider the $(d, \epsilon)$ code that, given an infinite list of codewords $z_1, z_2, \ldots$, outputs the first $d$-close match to $s$ as long as $s$ is not too atypical. Specifically,  the encoder outputs the lexicographic binary encoding (including the empty string) of 
\begin{equation}
W \triangleq 
\begin{cases}
\min \left\{ m \colon \mathsf d(S, z_m) \leq d \right\} &  \left \langle - \log_2 P_{Z}(B_d(S)) \right \rangle_\epsilon > 0 \\
 1 & \text{otherwise}
\end{cases}
\end{equation}

The encoded length averaged over both the source and all codebooks with codewords $Z_1, Z_2, \ldots$ drawn i.i.d. from $P_{Z}$ is upper bounded by   \begin{align}
 &~
 \E{\lfloor \log_2 W \rfloor } \notag\\
 \leq&~ \E{\log_2 W ~ \1{ \left \langle - \log_2 P_{Z}(B_d(S)) \right \rangle_\epsilon > 0 }} \label{-Alossya0}\\
 =&~ \E{\1{ \left \langle - \log_2 P_{Z}(B_d(S)) \right \rangle_\epsilon > 0 } \E{\log_2 W |S }} \\
 \leq&~ \E{ \1{ \left \langle - \log_2 P_{Z}(B_d(S)) \right \rangle_\epsilon > 0 } \log_2 \E{W | S}} \label{-Alossya} \\
 =&~ \E{ \left \langle - \log_2 P_{Z}(B_d(S)) \right \rangle_\epsilon }  \label{-Alossyb}
\end{align}
 where
\begin{itemize}
\item \eqref{-Alossya} is by Jensen's inequality;
\item  \eqref{-Alossyb} holds because conditioned on $S = s$ and averaged over codebooks, $W$ has geometric distribution with success probability $P_Z(B_d(s))$. 
\end{itemize}
It follows that there is at least one codebook that yields the encoded length not exceeding the expectation in \eqref{-Alossyb}. 

\end{proof}

\begin{remark}
Both  \eqref{eq:alon2lossyeps} and \eqref{Clossy} can be strengthened as in Remark \ref{rem:alon}.   
\end{remark}

\apxonly{ TODO: Show the counterpart of \eqref{Alossy0Heps} for $\epsilon > 0$. 
}



\begin{thm}[Bounds to $ {\mathbb R}_{S}(d, \epsilon)$ and to $H_{d, \epsilon}(S)$ ]
For all $d > d_{\min}$ we have
\begin{align}
&
\!\!\!\E{ \left \langle \jmath_{S}(S, d) \right \rangle_\epsilon}    - \log_2 \left( \mathbb R_S(d) - \mathbb R_S^\prime(d) d + 1 \right)  -  \log_2 e - h(\epsilon)  
\notag \\
\leq &~ {\mathbb R}_{S} (d, \epsilon) \label{minIlossylower}\\
\leq&~ {\mathbb R}_{S}^+(d, \epsilon)   \label{minIlossyupper}
\end{align}
and for all $d \geq d_{\min}$ we have
\begin{align}
&~{\mathbb R}_{S}^+(d, \epsilon) - \phi(\max\left\{1 - \epsilon, e^{-1} \right\})
\notag\\
 \leq&~ H_{d, \epsilon}(S)  \label{eq:Hdepslower} \\
\leq&~
{\mathbb R}_{S}^+(d, \epsilon) + \log_2 \left({\mathbb R}_{S}^+(d, \epsilon) + 1 + \phi\left(\min\left\{\epsilon, e^{-1} \right\} \right) \right) 
\notag\\
+&~ 1 + \phi\left(\min\left\{\epsilon, e^{-1} \right\} \right) \label{eq:Hdepsupper}
\end{align}
where $0 \leq \phi(\cdot) \leq e^{-1} \log_2 e $ is defined in \eqref{eq:phi}. 
\label{thm:minIlossy}
\end{thm}

\begin{proof}
Appendix \ref{appx:minIlossy}. 
\end{proof}

Trivially, ${\mathbb R}_{S} (d, \epsilon) \leq H_{d, \epsilon}(S)$. 

\begin{remark}
In the almost-lossless setting (Hamming distortion and $d = 0$),  the following bounds hold (Appendix \ref{appx:H0eps}). 
\begin{align}
&~\EE[\left \langle \imath_S(S) \right \rangle_\epsilon ] - \phi\left( \max\left\{1 - \epsilon, e^{-1} \right\}\right) \notag\\
\leq&~ H_{0, \epsilon}(S)  \label{eq:H0epslower} \\
\leq&~
\EE[\left \langle \imath_S(S) \right \rangle_\epsilon ]  +  \phi \left(\min\left\{\epsilon, e^{-1} \right\}\right) \label{eq:H0epsupper}
\end{align}
\end{remark}

\begin{remark}
 Particularizing \eqref{eq:Hdepslower} to the case $\epsilon = 0$, we recover the lower bound on $\epsilon$-entropy in \cite[Lemma 9]{posner1967epsilonentropy}: 
\begin{equation}
\inf_{P_Z}\E{ - \log_2 P_{Z}(B_d(S))  } \leq H_d(S) \label{eq:Hepslower}\\
\end{equation}
\end{remark}

\begin{remark}
 As follows from Lemma \ref{lemma:Rdlower1} in Appendix \ref{appx:minIlossy}, in the special case where
 \begin{equation}
 \jmath_S(S, d)  = \mathbb R_{S}(d) \text{ a.s.} \label{jsame}
\end{equation}
which in particular includes the equiprobable source under a permutation distortion measure (e.g. symbol error rate)\cite{dembo2001critical}, the lower bound in \eqref{minIlossylower} can be tightened as
\begin{equation}
{\mathbb R}_{S}(d, \epsilon) \geq  (1 - \epsilon) \mathbb R_{S}(d)  - h(\epsilon)  \label{Rdlower1}
\end{equation}
\end{remark}

\begin{remark}
Applying \eqref{eq:Xepsvar} to the random variable $\jmath_S(S, d)$, we have the variational characterization:
\begin{align} \EE[ \left \langle \jmath_S(S, d) \right \rangle_\epsilon] &= 	\mathbb R_S(d) - \max_{
\substack { \varepsilon \colon \mathcal M \mapsto [0, 1] \\ \EE[\varepsilon(S)] \le \epsilon}} \EE[\varepsilon(S) \jmath_S(S, d)] 
		\label{eq:lmb2lossy}
\end{align}
from where it follows, via \eqref{dtiltedupper}, that
\begin{align}
 \EE[ \left \langle \jmath_S(S, d) \right \rangle_\epsilon] &\leq  \E{ \left \langle - \log_2 P_{Z^\star}(B_d(S)) \right \rangle_\epsilon } \label{eq:dballlower}\\
 &\leq  \EE[ \left \langle \jmath_S(S, d) \right \rangle_\epsilon] +    \E{   - \log_2 P_{Z^\star}(B_d(S))  } \notag\\
 &- \mathbb R_S(d)  \label{eq:dballupper}
\end{align}
where $P_{Z^\star}$ is the output distribution that achieves $\mathbb R_S(d)$. 	
\end{remark}

\subsection{Asymptotic analysis}
\label{sec:2orderlossy}
In this section we assume that the following conditions are satisfied. 
\begin{enumerate}[(i)]
\item The source $\{S_i\}$ is stationary and memoryless,  $P_{S^k}  = P_{\mathsf S} \times \ldots \times P_{\mathsf S}$. \label{item:first}
\item The distortion measure is separable, $\mathsf d(s^k, z^k) = \frac 1 k \sum_{i = 1}^k \mathsf d(s_i, z_i)$. \label{item:separable}
\item The distortion level satisfies $d_{\min} < d < d_{\max}$, where $d_{\min}$ is defined in \eqref{dmin}, and $d_{\max} =\inf_{\mathsf z \in \widehat {\mathcal M} } \E{\mathsf d(\mathsf S, \mathsf z)}$, where the expectation is with respect to the unconditional distribution of $\mathsf S$. 
 \label{item:dminmax}
 \item $
 \E{{\mathsf d}^{12}(\mathsf S, \mathsf Z^\star)} < \infty
$  
where the expectation is with respect to $P_{\mathsf S} \times P_{\mathsf Z^\star}$, and $\mathsf Z^\star$ achieves the rate-distortion function  $\mathbb R_{\mathsf S}(d)$. \label{item:last}
\end{enumerate}
If conditions \eqref{item:first}--\eqref{item:dminmax} are satisfied, then $\lambda_{S^k} = k \lambda_{\mathsf S}$ and $P_{Z^{k \star} | S^k} = P_{\mathsf Z^\star | \mathsf S} \times \ldots \times P_{\mathsf Z^\star | \mathsf S}$, where $P_{\mathsf Z^\star | \mathsf S}$ achieves $\mathbb R_{\mathsf S}(d)$. Moreover, even if $\mathbb R_{\mathsf S}(d)$ is not achieved by any conditional distribution
\begin{align}
  \jmath_{S^k}(s^k, d) 
  &= \sum_{i = 1}^k \jmath_{\mathsf S}(s_i, d) \label{dtiltedadditive}
\end{align}
Finiteness of the twelfth moment of ${\mathsf d}(\mathsf S, \mathsf Z^\star)$ in restriction \eqref{item:last} is required for the achievability part of the asymptotic expansion in Theorem \ref{thm:minI2orderlossy}. 
\begin{thm}
Under assumptions \eqref{item:first}--\eqref{item:last}, for any $0 \leq \epsilon \leq 1$
\begin{small}
\begin{equation}
 \left.\begin{aligned}
        & L_{S^k}^\star(d, \epsilon)\\
        & {\mathbb R}_{S^k}(d, \epsilon)\\
        & {\mathbb R}_{S^k}^+(d, \epsilon)\\
         & H_{d, \epsilon}(S^k)\\
        & \E{ \left \langle \jmath_{S^k}(S^k, d) \right \rangle_\epsilon  }
       \end{aligned}
 \right\}
 = (1 - \epsilon) k R(d) -   \sqrt{\frac{k \mathcal V(d)}{2 \pi} } e^{- \frac { (\Qinv{\epsilon})^2} 2 }  + \theta(k)  \label{minI2orderlossy} 
\end{equation} 
\end{small}
where
\begin{align}
\mathcal V(d) &= \Var{\jmath_{\mathsf S}(\mathsf S, d)}
\end{align}
is the rate-dispersion function,
and the remainder term in the expansion satisfies
\begin{equation}
- 2 \log_2 k + \bigo{1}  \leq \theta(k) 
\leq \frac 3 2 \log_2 k + \bigo{1}
\end{equation}


\label{thm:minI2orderlossy}
\end{thm}

\begin{proof}
Due to \eqref{Jstarbounds}, the assumption \eqref{item:last} implies that the twelfth (and thus the third) moment of $\jmath_{\mathsf S}(\mathsf S, d)$ is finite, and the expansion for $\E{ \left \langle \jmath_{S^k}(S^k, d) \right \rangle_\epsilon  }$  follows from \eqref{dtiltedadditive} and Lemma \ref{lemma:EXeps}.  The converse direction is now immediate from Theorems \ref{thm:Elenlossy} and \ref{thm:minIlossy}. The achievability direction follows by an application of Lemma \ref{lemma:dispersionlossy} below to weaken the upper bounds in Theorems~\ref{thm:Elenlossy} and~\ref{thm:minIlossy}. 
\end{proof}

\begin{lemma}
 Let $0 \leq \epsilon \leq 1$. Under assumptions \eqref{item:first}--\eqref{item:last} 
\begin{align}
 \E{ \left \langle - \log_2 P_{Z^{ k \star} }(B_d(S^k)) \right \rangle_\epsilon  } 
 &= (1 - \epsilon) k R(d)  \label{lossyupper2order}\\
&-  \sqrt{\frac{k \mathcal V(d)}{2 \pi} } e^{- \frac { (\Qinv{\epsilon})^2} 2 } + \theta(k) \notag 
\end{align}
where
\begin{equation}
\bigo{1} \leq  \theta(k) \leq \frac 1 2 \log_2 k +  \bigo{1}
\end{equation}
\label{lemma:dispersionlossy}
\end{lemma}
\begin{proof}
 Appendix \ref{appx:dispersionlossy}. 
\end{proof}

\appendices

\section{Proof of Lemma \ref{lemma:EXeps}}
\label{appx:EXeps}
The following non-uniform strengthening of the Berry-Esse{\'e}n inequality is instrumental in the proof of Lemma \ref{lemma:EXeps}. 
\begin{thm}[Bikelis (1966), e.g. \cite{petrov1995limit}]
 Fix a positive integer $k$. Let $X_i$, $i = 1, \ldots, k$ be independent, $\E{X_i} = 0$, $\E{|X_i|^3} < \infty$. Then, for any real $t$
\begin{equation}
\left| \mathbb P \left[ \sum_{i = 1}^k X_i > t \sqrt {k V_k } \right]  - Q(t) \right| \leq \frac {B_k}{\sqrt k (1 + |t|^3)},
\label{sc:BerryEsseen}
\end{equation}
where
\begin{align}
V_k &= \frac 1 k \sum_{i = 1}^k \E{|X_i|^2} \label{sc:BerryEsseenVn}\\
T_k &= \frac 1 k \sum_{i = 1}^k \E{ |X_i|^3 } \label{sc:BerryEsseenTn}\\
B_k &= \frac{c_0 T_k}{V_k^{3/2}} \label{sc:BerryEsseenBn}
\end{align}
and $c_0$ is a positive constant. 
\label{thm:bikelis}
\end{thm}

Denote for brevity
\begin{equation}
 Y_k \triangleq \sum_{i = 1}^k X_i 
\end{equation}

If $\Var{\mathsf X} = 0$
\begin{equation}
\E{ \left \langle  Y_k \right \rangle_\epsilon } = (1 - \epsilon) k \E{\mathsf X},
\end{equation}
and \eqref{EXeps} holds.  

If $\Var{\mathsf X} > 0$ notice that
\begin{align}
&~ 
(1 - \epsilon) k \E{\mathsf X} - \E{ \left \langle Y_k \right \rangle_\epsilon } 
\label{eq:-Xepsa}  \\
=&~ \E{ \left( Y_k - k \E{\mathsf X} \right)1 \left\{ Y_k > \eta \right\}  } + \alpha \left(   \eta - k \E{\mathsf X} \right) \Prob{ Y_k = \eta }  \notag \\
=&~  \int_{\eta}^\infty \Prob{Y_k > t} dt+ \epsilon \left( \eta - k \E{\mathsf X} \right) \label{eq:-Xepsb},
\end{align}
where $\eta$ and $\alpha$ are those in \eqref{epscutoff}, and to write \eqref{eq:-Xepsb} we used 
\begin{align}
\E{Y_k \1{Y_k > \eta}}  &= \int_{\eta}^\infty \Prob{Y_k > t} dt + \eta \Prob{Y_k > \eta} \label{eq:-einfoa}.
\end{align}
We proceed to evaluate the right side of \eqref{eq:-Xepsb}. 
Using Theorem \ref{thm:bikelis}, we observe that \eqref{epscutoff} requires that $\eta$ behaves as
\begin{equation}
\eta = k \E{\mathsf X} + \sqrt{k \Var{\mathsf X}} \Qinv{\epsilon} + b_k  \label{eta2order}
\end{equation}
where $b_k = \bigo{1}$. Using \eqref{eta2order}, we may write
\begin{align}
&~  \int_{\eta}^\infty \Prob{Y_k > t} dt \notag\\
= &~ \int_0^{\infty} \Prob{ Y_k   > \eta + t} dt \label{eq:-einfoa} \\
= &~ \int_{b_k}^{\infty} \Prob{ Y_k   > k \E{\mathsf X} + \sqrt {k \Var{\mathsf X} } \Qinv{\epsilon}  + t } dt \\
= &~ \int_{0}^{\infty} \Prob{ Y_k   > k \E{\mathsf X} + \sqrt {k \Var{\mathsf X} } \Qinv{\epsilon}  + t } dt \notag\\
&+ \bigo{1}\\
= &~ \sqrt{k \Var{\mathsf X}}  \int_0^{\infty}Q\left( \Qinv{\epsilon} + r \right) dr + \bigo{1} \label{eq:-einfob}\\
= &~ \sqrt{k \Var{\mathsf X}} \int_{ \Qinv{\epsilon}}^{\infty} Q\left( r \right) dr + \bigo{1} \\
= &~ \sqrt{k \Var{\mathsf X}} \left[  \int_{ \Qinv{\epsilon}}^{\infty}  \frac 1 {\sqrt{2 \pi}} x e^{- \frac{x^2} {2}}dx  - \epsilon  \Qinv{\epsilon} \right] + \bigo{1} \\
=& \sqrt {k \Var{\mathsf X} }\left(  \frac{1}{\sqrt{2 \pi}} e^{- \frac { (\Qinv{\epsilon})^2} 2 }  - \epsilon \Qinv{\epsilon} \right) + \bigo{1} \label{eq:-einfoc}
\end{align}
where 
\eqref{eq:-einfob} follows by applying Theorem \ref{thm:bikelis} to the integrand in the left side and observing that 
\begin{equation}
\int_0^\infty \frac{dr}{ 1 + (\Qinv{\epsilon} + r)^3 } < \infty
\end{equation}

Applying \eqref{eta2order} and \eqref{eq:-einfoc} to \eqref{eq:-Xepsb}, we conclude that
\begin{align}
 (1 - \epsilon) k \E{\mathsf X} - \E{ \left \langle Y_k \right \rangle_\epsilon } =   \frac{\sqrt {k \Var{\mathsf X} }}{\sqrt{2 \pi}} e^{- \frac { (\Qinv{\epsilon})^2} 2 } + \bigo{1}, 
\end{align}
which is exactly \eqref{EXeps}.

\section{Proof of \eqref{disp0}}
\label{appx:disp0}

Denote for brevity
\begin{equation}
f(\epsilon) =  \frac 1 {\sqrt{2 \pi}}  e^{- \frac {\left( \Qinv{\epsilon} \right)^2} 2} 
\end{equation}
Direct computation yields
\begin{align}
f(\epsilon) &=  - \frac 1 {\left( Q^{-1} \right)^\prime \left( \epsilon \right) }\\
f^\prime (\epsilon) &= \Qinv{\epsilon} \label{dispa}\\
f^{\prime \prime}(\epsilon) &= - \frac 1 {f(\epsilon)}\label{dispa1}
\end{align}
Furthermore, using the bounds
\begin{align}
\frac{x}{\sqrt{2 \pi}(1 + x^2)} e^{-\frac {x^2}{2}} < Q(x) < \frac 1 {\sqrt{2 \pi}x } e^{-\frac {x^2}{2}},\, x > 0
\end{align}
we infer that as $\epsilon \to 0$
\begin{equation}
\Qinv{\epsilon} = \sqrt{2 \log_e \frac 1 \epsilon}  + \bigo{\log_e \log_e \frac 1 \epsilon} \label{dispb}
\end{equation}
Finally
\begin{align}
\lim_{\epsilon \to 0} \frac{f(\epsilon) - \epsilon \sqrt {2 \log_e \frac 1 \epsilon}}{ \epsilon} 
&= \lim_{\epsilon \to 0} \frac{f(\epsilon) - \epsilon f^\prime(\epsilon)}{\epsilon} \label{dispc}\\
&= \lim_{\epsilon \to 0} f^{\prime \prime}(\epsilon) \epsilon \label{dispd}\\
&= \lim_{\epsilon \to 0} \frac{-\epsilon}{f(\epsilon)}  \label{dispe}\\
&= \lim_{\epsilon \to 0} \frac{1}{\Qinv{\epsilon}} \label{dispf}\\
&= 0
\end{align}

where
\begin{itemize}
 \item \eqref{dispc} is due to \eqref{dispa} and \eqref{dispb}; 
 \item \eqref{dispd} is by the l'H\^opital rule; 
 \item \eqref{dispe} applies \eqref{dispa1}; 
 \item \eqref{dispf} is by the l'H\^opital rule and  \eqref{dispa}. 
\end{itemize}

\section{Proof of Theorem \ref{thm:minIlossy}}
\label{appx:minIlossy}

Given $P_S$, $\mathsf d$, denote for measurable $\mathcal F \subseteq \mathcal M$
\begin{equation}
\mathbb R_{S| \mathcal F} (d, \epsilon) \triangleq \min_{ \substack{ P_{Z | S } \colon \\ \Prob{ \mathsf d( S, Z) > d | S \in \mathcal F} \leq \epsilon}} I(S; Z | S \in \mathcal F) 
\end{equation}

In the proof of the converse bound in \eqref{minIlossylower}, the following result is instrumental.

\begin{lemma}
Suppose $P_S$, $\mathsf d$, $d > d_{\min}$ and $\mathcal F \subseteq \mathcal M$ are such that for all $s \in \mathcal F$
\begin{equation}
 \jmath_S(S, d)  \geq r \text{ a.s.} \label{jlower}
\end{equation}
for some real $r$. 
Then
\begin{equation}
\mathbb R_{S| \mathcal F} (d, \epsilon) \geq \left| (1 - \epsilon) r  + (1 - \epsilon) \log_2 \Prob{S \in \mathcal F} - h(\epsilon) \right|^+ \label{Rdlower1}
\end{equation}
\label{lemma:Rdlower1}
\end{lemma}
\begin{proof}
Denote
\begin{align}
p_S(z) &\triangleq \Prob{ \mathsf d(S, z) \leq d | S \in \mathcal F}\\
p &\triangleq \sup_{z \in \widehat{\mathcal M}} p_S(z) 
\end{align}

If $\epsilon > 1 - p$, ${\mathbb R}_{S}(d, \epsilon)$ = 0, so in the sequel we focus on the nontrivial case 
\begin{equation}
\epsilon \leq 1 - p \label{epsub}
\end{equation}

To lower-bound the left side of \eqref{Rdlower1}, we weaken the supremum in \eqref{RR(d)csiszar} by selecting a suitable pair $(J(s), \lambda)$ satisfying the constraint in \eqref{csiszarg}. Specifically, we choose
\begin{align}
\exp(-\lambda) &= \frac{\epsilon p} {(1 - \epsilon) (1 - p)} \label{lambdachoice}\\
\exp(J(s))  &=\exp( J ) \triangleq   \frac{1 - \epsilon }{p} 
 , ~ s \in \mathcal F \label{Jchoice}
\end{align}
To verify that the condition \eqref{csiszarg} is satisfied, we substitute the choice in \eqref{lambdachoice} and \eqref{Jchoice}  into the left side of  \eqref{csiszarg} to obtain
\begin{align}
&~\epsilon \frac{ 1 - p_S(z)}{1 - p}  + (1 - \epsilon)  \frac{p_S(z)}{p} \notag\\
\leq&~ (1 - p) \left[ \frac{ 1 - p_S(z)}{1 - p}  -  \frac{p_S(z)}{p}  \right] +  \frac{p_S(z)}{p} \label{epsub1}\\
=&~ 1
\end{align}
where \eqref{epsub1} is due to \eqref{epsub} and the observation that the expression in square brackets in the right side of \eqref{epsub1} is nonnegative. 
Plugging  \eqref{lambdachoice} and \eqref{Jchoice} into \eqref{RR(d)csiszar}, we conclude that
\begin{align}
\mathbb R_{S| \mathcal F} (d, \epsilon) &\geq J - \lambda \epsilon\\
 &= d( \epsilon \| 1 - p) \\
 &\geq (1 - \epsilon) \log_2 \frac 1 {p}  - h(\epsilon) \\
 &\geq (1 - \epsilon) r  + (1 - \epsilon) \log_2 \Prob{S \in \mathcal F} - h(\epsilon) \label{-Rdlower1a}
\end{align}
where $d(a \| b) = a \log \frac a b + (1 - a) \log \frac {1 - a}{1 - b}$ is the binary relative entropy function, and \eqref{-Rdlower1a} is due to 
\begin{align}
p_S(z) &\leq \E{ \exp(\lambda_S d - \lambda_S \mathsf d(S, z))  | S \in \mathcal F} \label{markov}\\
&\leq \E{ \exp(\jmath_S(S, d) + \lambda_S d - \lambda_S \mathsf d(S, z) - r)  | S \in \mathcal F} \label{-jlower}\\
&\leq \frac{\exp(-r)}{\Prob{S \in \mathcal F}} \E{ \exp(\jmath_S(S, d) + \lambda_S d - \lambda_S \mathsf d(S, z) ) } \\
&\leq \frac{\exp(-r)}{\Prob{S \in \mathcal F}}  \label{-Rdlower1b}
\end{align}
where $\lambda_S \triangleq -\mathbb R_{S}(d)$, and
\begin{itemize}
 \item \eqref{markov} is Markov's inequality;
 \item \eqref{-jlower} applies \eqref{jlower}; 
 \item \eqref{-Rdlower1b} is equivalent to \eqref{csiszarg1}. 
\end{itemize}
\end{proof}

\begin{proof}[Proof of Theorem \ref{thm:minIlossy}]
We start with the converse bound in \eqref{minIlossylower}. Note first that, similar to \eqref{eq:itt5a}, the constraint in \eqref{RR(eps)} is achieved with equality. Denoting the random variable
\begin{equation}
 F \triangleq \left \lfloor \jmath_S(S, d) \right \rfloor + 1
\end{equation}
and the sets
\begin{equation}
\mathcal F_j \triangleq \left\{s \in \mathcal M \colon F = j\right\}, 
\end{equation}
we may write 
\begin{align}
I(S; Z) &= I(S, F; Z)\\
&=  I(S; Z | F) + I(F; Z) 
\end{align}
so
\begin{align}
 {\mathbb R}_{S} (d, \epsilon)
&\geq \min_{ \substack{ P_{Z| S} \colon \\ \Prob{\mathsf d(S, Z) > d } \leq \epsilon} } I(S; Z | F) \\
&= \min_{\varepsilon(\cdot) \colon \E{\varepsilon(F)} \leq \epsilon} \sum_{j = -\infty}^\infty P_F(j) \mathbb R_{S| \mathcal F_j} (d, \epsilon(j))  \label{minIlbalossy}
\end{align}
We apply Lemma \ref{lemma:Rdlower1} to lower bound each term of the sum by
\begin{align}
&~
 \mathbb R_{S| \mathcal F_j} (d, \epsilon(j))
 \notag\\
 \geq&~ \left| (1 - \epsilon(j)) j  + (1 - \epsilon) \log_2 P_F(j) - h(\epsilon(j)) \right|^+ \label{-minIlbblossy}
\end{align}
to obtain
\begin{align}
&~ {\mathbb R}_{S} (d, \epsilon) \notag\\
\geq&~ \min_{\varepsilon(\cdot) \colon \E{\varepsilon(F)} \leq \epsilon}\left\{  \E{(1 - \epsilon(F))  \jmath_{S} (S, d) } -  \E{h(\epsilon(F))} \right\} \notag\\
&- H(F)  \label{minIlbblossy} \\
=&~ \min_{\varepsilon(\cdot) \colon \E{\varepsilon(F)} \leq \epsilon} \left\{ \E{(1 - \epsilon(F))  \jmath_{S} (S, d) }  \right\}  -  H(F)  - h(\epsilon) \label{minIlbb0lossy}\\
\geq &~
\E{ \left \langle \jmath_{S}(S, d) \right \rangle_\epsilon} -  H(F)  - h(\epsilon) \label{minIlbclossy}\\
\geq &~ \E{ \left \langle \jmath_{S}(S, d) \right \rangle_\epsilon}    - \log_2 \left( \E{J_S(S)} + 1 \right)  -  \log_2 e - h(\epsilon) \label{minIlbdlossy}
\end{align}
where \eqref{minIlbblossy} uses \eqref{dtiltedupper}, \eqref{minIlbb0lossy} is by concavity of $h(\cdot)$, \eqref{minIlbclossy} is due to \eqref{eq:lmb2lossy}, and \eqref{minIlbdlossy} holds  because $F + \lambda_S d \geq J_S(S) \geq 0$, and  the entropy of a random variable on $\mathbb Z_+$ with a given mean is maximized by that of the geometric distribution.

\apxonly
{
TODO: FIX. To tighten the logarithmic term in \eqref{minIlbdlossy}, write 
\begin{equation}
 {\mathbb R}_{S} (d, \epsilon)
\geq \min_{ \substack{ P_{Z| S} \colon \\ \Prob{\mathsf d(S, Z) > d } \leq \epsilon} } I(S; Z | F)  +  \min_{ \substack{ P_{Z| S} \colon \\ \Prob{\mathsf d(S, Z) > d } \leq \epsilon} } I(F; Z ) \label{minIlbalossy1}
\end{equation}

To lower-bound the second term in \eqref{minIlbalossy1}, we introduce 
\begin{equation}
 E \triangleq 1 \left\{ \mathsf d(S, Z) > d \vee \jmath_S(S, d) \leq -1 \right\}
\end{equation}
and
\begin{equation}
p_E \triangleq \Prob{E =1} 
\end{equation}
and we reason similarly to \eqref{eq:itv1}--\eqref{eq:itv6}:
	\begin{align} H(F|Z) &= I(F; E|Z) + H(F|Z,E)\\
		&\le h\left( p_E \right) + (1 - p_E) H(F|Z,E = 0) + p_E H(F|Z,E = 1)
\end{align}
As in \eqref{eq:itv1}--\eqref{eq:itv6}
\begin{equation}
 H(F| Z, E = 1 ) \leq \log_2 \left( 1 + \frac{\mathbb R_S(d)}{p_E} \right) + \log_2 e
\end{equation}

To bound $H(F|Z,E = 0)$, observe (TODO: careful with the logarithm base): 
\begin{align}
&~ \E{\jmath_S(S, d) 1 \left\{ \mathsf d(S, z) \leq d, \jmath_S(S, d) > - 1, Z = z \right\} } 
\notag \\
   \leq&~ \E{ \exp\left( \jmath_S(S, d)  \right) 1 \left\{ \mathsf d(S, z) \leq d, Z = z \right\} 1\left\{ \jmath_S(S, d) > - 1\right\}} 
  \notag \\
  -&~ \E{ 1 \left\{ \mathsf d(S, z) \leq d, Z = z \right\} 1\left\{ \jmath_S(S, d) > - 1\right\}}\\
 \leq&~ \E{ \exp\left( \jmath_S(S, d) - \lambda_S \mathsf d(S, z) + \lambda_S d  \right) } 
 \notag\\
 -&~ \Prob{E = 1, Z = z}\\
 \leq&~ 1 - \Prob{E = 1, Z = z}
\end{align}

so  
\begin{align}
H(F | E = 0, Z = z) &\leq \log_2 \left( 1 +  \frac{\E{\jmath_S(S, d) 1 \left\{ E = 0, Z = z \right\} } }{\Prob{E = 0, Z = z} }\right) + \log_2 e\\
&\leq \log_2 \left( 1 +  \frac{1 - \Prob{E = 1, Z = z} }{\Prob{E = 0, Z = z}}\right) + \log_2 e
\end{align}

Take expectation, apply concavity of logarithm:
\begin{align}
H(F | Z, E = 0) \leq \log_2 \left( 1 + \frac {|\mathcal Z|  - p_E} {1 - p_E}  \right) + \log_2 e
\end{align}

Too loose!
}


 To show the upper bound in \eqref{minIlossyupper}, fix an arbitrary distribution $P_{\bar Z}$ and define the conditional probability distribution $P_{Z|S}$ through\footnote{Note that in general $P_S \to P_{Z|S} \nrightarrow P_{\bar Z}$.}
\begin{equation}
\frac{dP_{Z|S = s}(z)}{dP_{\bar Z}(z)} = 
\begin{cases}
  \frac{\1{ \mathsf d(s, z) \leq d}}{P_{ \bar Z}(B_d(s))} & \left \langle - \log_2 P_{\bar Z}(B_d(s)) \right \rangle_\epsilon > 0\\
  1 & \text{otherwise}
\end{cases}  \label{PZ|Slossy}
\end{equation}

By the definition of $P_{Z|S}$
\begin{equation}
\Prob{\mathsf d(S, Z) > d} \leq \epsilon  
\end{equation}
Upper-bounding the minimum in \eqref{RR(eps)} with the choice of $P_{Z|S}$ in \eqref{PZ|Slossy}, we obtain the following nonasymptotic bound:
\begin{align}
{\mathbb R}_{S}(d, \epsilon) 
\leq&~ 
I(S; Z) 
\\
=&~ 
D\left( P_{Z | S} \| P_{ \bar Z} | P_{S} \right) - D(P_{Z} \| P_{ \bar Z})\\
\leq&~ D\left( P_{Z | S} \| P_{\bar Z } | P_{S} \right) \\
 =&~ \E{ \left \langle - \log_2 P_{\bar Z}(B_d(S)) \right \rangle_\epsilon } 
\end{align}
which leads to \eqref{minIlossyupper} after minimizing the right side over all $P_{\bar Z}$.

To show the lower bound on $\left(\epsilon, \delta\right)$-entropy in \eqref{eq:Hdepslower}, 
fix $\mathsf f$ satisfying the constraint in \eqref{eq:Hepsdelta}, denote
\begin{align}
Z &\triangleq \mathsf f(S)\\
\varepsilon(s) &\triangleq 1 \left\{ \mathsf d(s, \mathsf f(s)) > d\right\} 
\end{align}
and write
\begin{align}
H(Z) &\geq H(Z | \varepsilon(S))\\
&\geq P_{\varepsilon(S)}(0) H(Z | \varepsilon(S) = 0)\\
&= \E{ \imath_{Z, \varepsilon(S) = 0} (Z) (1 - \varepsilon(S)) }  \notag\\
&+ P_{\varepsilon(S)}(0) \log_2 P_{\varepsilon(S)}(0)\\
&\geq  \E{ \left \langle - \log_2 P_{Z }(B_d(S)) \right \rangle_\epsilon } - \phi(\min\{\epsilon, e^{-1}\})
\end{align}
where the second term is bounded by maximizing $p \log_2 \frac 1 p$ over $[1 - \epsilon, 1]$, and the first term is bounded via the following chain. 
\begin{align}
&~\E{ \imath_{Z, \varepsilon(S) = 0} (Z) (1 - \varepsilon(S)) } \notag\\
\geq&~ \E{- \log_2 P_{Z}(B_d(S)) (1 - \varepsilon(S)) } \label{posnerg}\\
\geq&~ \min_{\varepsilon(\cdot) \colon \E{\varepsilon(S)} \leq \epsilon} \E{- \log_2 P_{Z}(B_d(S)) (1 - \varepsilon(S))}\\
=&~ \E{ \left \langle - \log_2 P_{Z }(B_d(S)) \right \rangle_\epsilon } \label{eq:Pepsvar}
\end{align}
where \eqref{posnerg} holds because due to  $\{s \in \mathcal M \colon \mathsf f(s) = z, \epsilon(s) = 0 \} \subseteq B_d(s)$ we have for all $s \in \mathcal M$
\begin{equation}
\Prob{Z = \mathsf f(s), \varepsilon(S) = 0} \leq P_{Z}(B_d(s))
\end{equation}
and \eqref{eq:Pepsvar} is due to \eqref{eq:Xepsvar}.

To show the upper bound on $\left(\epsilon, \delta\right)$-entropy in \eqref{eq:Hdepsupper}, fix $P_Z$ such 
\begin{equation}
P_{Z}(B_d(s)) > 0 
\end{equation}
for $P_S$-a.s. $s \in \mathcal M$, let $Z^\infty \sim P_{Z} \times P_{Z} \times \ldots$, and define $W$ as   
\begin{equation}
W \triangleq 
\begin{cases}
\min \left\{ m \colon \mathsf d(S, Z_m) \leq d \right\} &  \left \langle - \log_2 P_{Z}(B_d(S)) \right \rangle_{\epsilon^\prime} > 0 \\
 1 & \text{otherwise}
\end{cases}
\end{equation}
where $\epsilon^\prime$ is the maximum of $ \epsilon^\prime \leq \epsilon$ such that the randomization on the boundary of $\left \langle - \log_2 P_{Z}(B_d(S)) \right \rangle_{\epsilon^\prime}$ can be implemented without the actual randomization (see Section \ref{sec:optimal} for an explanation of this phenomenon). 

If $z_1, z_2, \ldots$ is a realization of $Z^\infty$, $\mathsf f(s) = z_w$ is a deterministic mapping that satisfies the constraint in \eqref{eq:Hepsdelta}, so, since $w \mapsto z_w$ is injective, we have
\begin{align}
H_{d, \epsilon}(S) \leq H(W | Z^\infty = z^\infty) 
\end{align}
 
We proceed to show that $H(W | Z^\infty) $ is upper bounded by the right side of \eqref{eq:Hdepsupper}. Via the random coding argument this will imply that there exists at least one codebook $z^\infty$ such that $H(W | Z^\infty = z^\infty) $ is also upper bounded by the right side of \eqref{eq:Hdepsupper}, and the proof will be complete.

Let
\begin{align}
G \triangleq \lfloor \log_2 W \rfloor \left \langle - \log_2 P_{Z}(B_d(S)) \right \rangle_{\epsilon^\prime} > 0
\end{align}
and consider the chain
\begin{align}
 H(W | Z^\infty) &\leq H(W) \label{eq:Hdepsuppera1}\\
 &= H(W | G)  + I(W; G) \\
 &\leq \E{G} + H(G) \label{eq:Hdepsupperb1}\\
 &\leq \E{G} + \log_2 \left( 1 + \E{G}\right) + \log_2 e \label{eq:Hdepsupperc1}
\end{align}
where
\begin{itemize}
\item  \eqref{eq:Hdepsuppera1} holds because conditioning decreases entropy; 
\item \eqref{eq:Hdepsupperb1} holds because conditioned on $G = i$, $W$ can have at most $i$ values; 
\item  \eqref{eq:Hdepsupperc1} holds because the entropy of a positive integer-valued random variable with a given mean is maximized by the geometric distribution. 
\end{itemize}
Finally, it was shown in \eqref{-Alossyb} that
\begin{align}
\E{G} &=  \E{\left \langle - \log_2 P_{Z}(B_d(S)) \right \rangle_{\epsilon^\prime}} \\
&\leq  \E{\left \langle - \log_2 P_{Z}(B_d(S)) \right \rangle_{\epsilon}} +\phi( \min\{\epsilon, e^{-1}\})
\end{align}
where $\phi(\cdot)$ is the no-randomization penalty as explained in the proof of \eqref{eq:Lstardeterm}.

\apxonly{
ALTERNATIVE PROOF OF A SLIGHTLY WEAKER BOUND. 
To show the upper bound on $\left(\epsilon, \delta\right)$-entropy in \eqref{eq:Hdepsupper}, 
 combine \eqref{eq:Lstarlossydeterm}, \eqref{eq:alon2lossyeps} and \eqref{Alossy} to obtain
\begin{align}
H_{d, \epsilon}(S) - \log_2 (H_{d, \epsilon}(S)+1) 
&\leq  {\mathbb R}_{S}^+(d, \epsilon) + (1 + e^{-1}) \log_2 e \label{eq:Hdepsuppera}
\end{align}
and observe that 
\begin{equation}
 x - \log_2 (x+1) \leq a \label{eq:Hdepsupperb}
\end{equation}
implies 
\begin{equation}
 x  \leq a + \log_2(a + 1) + \frac{\log_2 e \cdot \log_2(a + 1) }{a + 1 - \log_2 e} 
 \label{eq:Hdepsupperc} 
\end{equation}
for all $x \geq 0$, $a \geq \log_2 e$. Indeed, by concavity of the logarithm 
\begin{equation}
\log_2 (x+1) \leq  \log_2 (a + 1) + \frac{\log_2 e}{a + 1}(x - a) \label{eq:Hdepsupperd}
\end{equation}
so \eqref{eq:Hdepsupperb} can be weakened as 
\begin{align}
x - \log_2 (a + 1) - \frac{\log_2 e}{a + 1}(x - a) \leq a
\end{align}
which is equivalent to \eqref{eq:Hdepsupperc}. 

Denote 
\begin{equation}
 r \triangleq  {\mathbb R}_{S}^+(d, \epsilon) + (1 + e^{-1}) \log_2 e
\end{equation}
Applying \eqref{eq:Hdepsupperc} to \eqref{eq:Hdepsuppera}, we obtain the nonasymptotic bound
\begin{equation}
 H_{d, \epsilon}(S)\leq r + \log_2 (r +1 )  + \frac{\log_2 e \cdot \log_2(r + 1) }{r + 1 - \log_2 e}  \label{eq:Hdepsuppere}
\end{equation}
To obtain the simpler \eqref{eq:Hdepsupper}, apply \eqref{eq:Hdepsupperd} with $a = {\mathbb R}_{S}^+(d, \epsilon)$ to upper bound the second term in \eqref{eq:Hdepsupperd}, and use monotonicity of  $\frac{\log_2(r + 1) }{r + 1 - \log_2 e}$ on $r \geq \log_2 e$ to upper bound the third term in \eqref{eq:Hdepsuppere}. 

We have shown 
\begin{align}
H_{d, \epsilon}(S)  
&\leq {\mathbb R}_{S}^+(d, \epsilon) + \log_2 ({\mathbb R}_{S}^+(d, \epsilon) + 1 )  + C \label{eq:Hdepsupper}
\end{align}
where  
$C = (2 (1 + e^{-1}) + \log_2 (\log_2 e + 1)) \log_2 e \approx 5.81$~bits. 
}

\end{proof}

\section{Proof of the bounds \eqref{eq:H0epslower} and \eqref{eq:H0epsupper} on $H_{0, \epsilon}(S)$ (Hamming distortion)}
\label{appx:H0eps}

The upper bound in \eqref{eq:H0epsupper} is obtained by a suboptimal choice (in \eqref{eq:Hepsdeltalossless}) of $\mathsf f(s) = s$ for all $s \leq m_0$, where $m_0$ is that in \eqref{eq:m0}, and $\mathsf f(s) = m_0 + 1$ otherwise. 

To show the lower bound in \eqref{eq:H0epslower}, fix $\mathsf f$ satisfying the constraint in \eqref{eq:Hepsdeltalossless}, put 
\begin{align}
 \varepsilon(S) &\triangleq \1{S \neq \mathsf f(S)}
\end{align}
and write
\begin{align}
 H(\mathsf f(S)) &\geq H(\mathsf f(S) | \varepsilon(S) = 0) P_{\varepsilon(S)}(0) \label{eq:H0epsuppera} \\
 &= \E{ \log_2 \frac 1 {P_{\mathsf f(S) | \varepsilon(S) = 0}(S)} | \varepsilon(S) = 0 } P_{\varepsilon(S)}(0) \label{eq:H0epsupperb}\\
 &\geq H \left( S | \varepsilon(S) = 0 \right) P_{\varepsilon(S)}(0) \\
 &=  \E{\imath_S(S)  \1{\varepsilon(S) = 0 }} + P_{\varepsilon(S)}(0) \log_2 P_{\varepsilon(S)}(0)\\
 &\geq  \EE[\left \langle \imath_S(S) \right \rangle_\epsilon ] - \phi\left( \max\left\{1 - \epsilon, e^{-1} \right\}\right) \label{eq:H0epsupperc}
\end{align}
where 
\begin{itemize}
\item \eqref{eq:H0epsuppera} is because conditioning decreases entropy; 
\item \eqref{eq:H0epsupperb} is due to
\begin{equation}
 \min_{P_Y} \E{\imath_{Y}(X)} = H(X);
\end{equation}
\item in \eqref{eq:H0epsupperc}, the first term is bounded using \eqref{eq:Xepsvar}, and the second term is bounded by maximizing $p \log_2 \frac 1 p$ over $[1 - \epsilon, 1]$. 
\end{itemize}

\section{Proof of Lemma \ref{lemma:dispersionlossy}}
\label{appx:dispersionlossy}

The following refinement of the lossy AEP is essentially contained in \cite{yang1999redundancy}. 


\begin{lemma} Under restrictions \eqref{item:first}--\eqref{item:last}, there exist constants $C_1, C_2$ such that
eventually, almost surely
\begin{align}
 \log_2 \frac 1 {P_{Z^{k\star}}(B_d(S^k))} 
\leq&~ \sum_{i = 1}^k \jmath_{\mathsf S}(S_i, d) + \frac 1 2 \log_2 k + C_2 \label{eq:yang}
 \\
 -&~ k \lambda_{\mathsf S} (d - \bar {\mathsf d}(S^k)) + k C_1 (d - \bar {\mathsf d}(S^k))^2 \notag
\end{align}
where
\begin{equation}
\bar {\mathsf d}(s^k) \triangleq \frac 1 k \sum_{i = 1}^k \E{\mathsf d(s_i, \mathsf Z^\star) | \mathsf S = s_i} 
\end{equation}
\label{lemma:yang}
\end{lemma}
\begin{proof}
It follows from  {\cite[(4.6), (5.5)]{yang1999redundancy}} that the probability of violating \eqref{eq:yang} is $\bigo{\frac 1 {k^2}}$. Since $\sum_{k = 1}^\infty \frac 1 {k^2}$ is summable, by the Borel-Cantelli lemma  \eqref{eq:yang} holds w. p. 1 for $k$ large enough. 
\end{proof}

 Noting that  $\bar {\mathsf d}(s^k)$ is a normalized sum of independent random variables with mean $d$, we conclude using Lemma \ref{lemma:yang} that for $k$ large enough
\begin{equation}
\E{  \log_2 \frac 1 {P_{Z^{k\star}}(B_d(S^k))}} \leq k R(d) + \frac 1 2 \log_2 k + \bigo{1}
\end{equation}
 
  Lemma \ref{lemma:dispersionlossy} is now immediate 
from \eqref{eq:dballlower} and \eqref{eq:dballupper} and the expansion for $\E{ \left \langle \jmath_{S^k}(S^k, d) \right \rangle_\epsilon  }$ in \eqref{minI2orderlossy}.

\bibliographystyle{IEEEtran}
\bibliography{../rateDistortion,../reports}

\begin{IEEEbiographynophoto}{Victoria Kostina}(S'12-M'14)
 joined Caltech as an Assistant Professor of Electrical Engineering in the fall of 2014.
She holds a Bachelor's degree from Moscow institute of Physics and Technology (2004), where she was affiliated with the Institute for Information Transmission Problems of the Russian Academy of Sciences, a Master's degree from University of Ottawa (2006), and a PhD from Princeton University (2013). Her PhD dissertation on information-theoretic limits of lossy data compression received Princeton Electrical Engineering Best Dissertation award. 

Victoria Kostina's research spans information theory, coding, and wireless communications. Her current efforts explore 
the nonasymptotic regime in information theory. 
\end{IEEEbiographynophoto}

\begin{IEEEbiographynophoto}{Yury Polyanskiy}(S'08-M'10-SM'14)
 is an 
Associate Professor of Electrical Engineering and Computer Science and a member of LIDS at MIT.
Yury received the M.S. degree in applied mathematics and physics from the Moscow Institute of Physics and Technology,
Moscow, Russia in 2005 and the Ph.D. degree in electrical engineering from Princeton
University, Princeton, NJ in 2010. In 2000-2005 he lead the development of the embedded software in the
Department of Surface Oilfield Equipment, Borets Company LLC (Moscow). Currently, his research focuses on basic questions in information theory, error-correcting codes, wireless communication and fault-tolerant and defect-tolerant circuits.
Dr. Polyanskiy won the 2013 NSF CAREER award and 2011 IEEE Information Theory Society Paper Award.
\end{IEEEbiographynophoto}

\begin{IEEEbiographynophoto}
{Sergio Verd\'{u}} (S'80-M'84-SM'88-F'93) received the Telecommunications Engineering degree from the
Universitat Polit\`{e}cnica de Barcelona in 1980, and the Ph.D. degree in Electrical Engineering from the
University of Illinois at Urbana-Champaign in 1984. Since 1984 he has been a member of the faculty of
Princeton University, where he is the Eugene Higgins Professor of Electrical Engineering, and is a member
of the Program in Applied and Computational Mathematics.

Sergio Verd\'{u} is  the recipient of the 2007 Claude E. Shannon Award, and
the 2008 IEEE Richard W. Hamming Medal. 
He is a member of both the National Academy of Engineering and the National Academy of Sciences.

Verd\'{u} is a recipient of several paper awards from the IEEE: 
the 1992 Donald Fink Paper Award, 
the 1998 and 2012 Information Theory  Paper Awards, 
an Information Theory Golden Jubilee Paper Award,
the 2002 Leonard Abraham Prize Award,  
the 2006 Joint Communications/Information Theory Paper Award, 
and the 2009 Stephen O. Rice Prize from the IEEE Communications Society.  
In 1998, Cambridge University Press published his book {\em Multiuser Detection,} 
for which he received the 2000 Frederick E. Terman Award from the American Society for Engineering Education. 
He was awarded a Doctorate Honoris Causa from the Universitat  Polit\`{e}cnica de Catalunya in 2005.

Sergio Verd\'{u} served as President of the IEEE Information Theory Society in 1997, and
on its Board of Governors (1988-1999, 2009-2014).
He has also served in various editorial capacities for the {\em IEEE Transactions on Information Theory}:
Associate Editor (Shannon Theory, 1990-1993; Book Reviews, 2002-2006),  
Guest Editor of the Special Fiftieth Anniversary Commemorative Issue
(published by IEEE Press as ``Information Theory: Fifty years of discovery"), 
and member of the Executive Editorial Board (2010-2013).
He is the founding Editor-in-Chief of {\em Foundations and Trends in Communications and Information Theory}. 
Verd\'{u} is co-chair of the {\em 2016 IEEE International Symposium on Information Theory}, which will take place in his hometown.
\end{IEEEbiographynophoto}

\end{document}